%% file: paper.tex
\documentclass{article}
\usepackage{mathpartir}
\usepackage{times}  
\usepackage{amsmath}
\usepackage{amsthm}
\input{nominal}
\newtheorem{theorem}{Theorem}[section]
\newtheorem{proposition}[theorem]{Proposition}
\newtheorem{lemma}[theorem]{Lemma}
\newtheorem{corollary}[theorem]{Corollary}
\theoremstyle{remark}
\newtheorem{remark}[theorem]{Remark}
\theoremstyle{definition}

\newtheorem{definition}[theorem]{Definition}

\renewcommand{\phi}{\varphi}

\newcommand{\FLseq}{\ensuremath{ FL_{Seq}}\xspace}
\newcommand{\FL}{\ensuremath{FL}\xspace}
\newcommand{\NL}{\ensuremath{NL}\xspace}
\renewcommand{\bar}[1]{\overline{#1}}
\renewcommand{\vec}{\bar}
\newcommand{\propp}{o}

\newcommand{\Var}{\mathbb{V}}
\newcommand{\Name}{\mathbb{A}}
\begin{document}

\title{A Simple Sequent Calculus for Nominal Logic}
\author{James Cheney}
\maketitle
\begin{abstract}
  Nominal logic is a variant of first-order logic that provides
  support for reasoning about bound names in abstract syntax.  A key
  feature of nominal logic is the new-quantifier, which quantifies
  over \emph{fresh names} (names not appearing in any values
  considered so far).  Previous attempts have been made to develop
  convenient rules for reasoning with the new-quantifier, but we argue
  that none of these attempts is completely satisfactory.

  In this article we develop a new sequent calculus for nominal logic in
  which the rules for the new-quantifier are much simpler than in previous
  attempts.  We also prove several structural and metatheoretic
  properties, including cut-elimination, consistency, and
 equivalence to Pitts' axiomatization of nominal
  logic.
\end{abstract}

\section{Introduction}

Nominal logic~\cite{pitts03ic} is a variant of first-order logic with
additional constructs for dealing with \emph{names} and \emph{binding}
(or \emph{name-abstraction}) based on the primitive notions of
bijective renaming (\emph{swapping}) and name-independence
(\emph{freshness}).  It was introduced by Pitts~\cite{pitts03ic} as a
first-order and reasonably well-behaved fragment of
\emph{Fraenkel-Mostowski set theory}, the setting for Gabbay and
Pitts' earlier foundational work on formalizing names, freshness, and
binding using swapping~\cite{gabbay02fac}.

One of the most interesting features of nominal logic is the presence
of a novel form of quantification over \emph{fresh names}.  The
formula $\new \Aa.\phi$ means, intuitively, ``for fresh names
$\Aa$, 
$\phi$ holds''.  The intended semantics of nominal logic interprets
expressions as values in \emph{finitely-supported nominal sets}, or
sets acted upon by name-swapping and such that each value depends on
at most finitely many names.  The inspiration for the
$\new$-quantifier is the fact that in the presence of infinitely many
names, a fresh name can be chosen for any finitely-supported value,
and equally-fresh names are indistinguishable.  As a result, a
property $\phi(a)$ holds for \emph{some} fresh name $a$ if and only if
it holds for \emph{all} fresh names; in either case, we say that $\new
a.\phi$ holds.

Several formalizations of nominal logic have been investigated.  Pitts
introduced nominal logic as a Hilbert-style axiomatic system.
Gabbay~\cite{gabbay07jal} proposed Fresh Logic (\FL), an
intuitionistic Gentzen-style natural deduction system.  Gabbay and
Cheney~\cite{gabbay04lics} presented \FLseq, a sequent calculus
version of Fresh Logic.  Sch\"opp and Stark have developed a dependent
type theory of names and binding that contains nominal logic as a
special case~\cite{schoepp04csl}.

However, none of these formalizations is ideal.  Hilbert systems have
well-known deficiencies for computer science applications.  \FL and
\FLseq rely on a complicated technical device called \emph{slices} for
the rules involving $\new$.  Sch\"opp and Stark's system is much more
powerful than seems necessary for many applications of nominal logic,
and there are many unresolved issues, such as proof normalization and
the decidability of the equality and typechecking judgments.

In this article we present a new and simpler sequent calculus for
nominal logic.  Its main novelty is the use of freshness information
in typing contexts needed in reasoning
about $\new$-quantified formulas, rather than the technically more
cumbersome \emph{slices} used in \FL and \FLseq.  We prove basic
proof-theoretic results such as cut-elimination, establishing that
this calculus is proof-theoretically sensible.  In addition, we prove
that \NLseq is consistent and equivalent to Pitts' original
axiomatization of nominal logic.

This article incorporates some revised material from a previous
conference publication~\cite{cheney05fossacs}, extended with detailed
proofs and additional results concerning conservativity. That paper
also gave a sound and complete embedding of Miller and Tiu's
\FOLNabla~\cite{miller05tocl} in \NLseq, extending an
earlier result by Gabbay and Cheney~\cite{gabbay04lics} which gave a
sound, but nonconservative translation from \FOLNabla to \FLseq.
These results are not presented in this article.

\section{Background}

\subsection{Pitts' axiomatization}\labelSec{background-nl}

As presented by Pitts, nominal logic consists of typed first-order
logic with equality and with a number of special types, type
constructors, and function and relation symbols formalized by a
collection of axioms.  In particular, the basic sort symbols of
nominal logic are divided into \emph{data types} $\delta,\delta'$ and
\emph{atom types} $\nu,\nu'$ (which we shall also preferentially call
\emph{name types}).  In addition, whenever $\nu$ is a name type and
$\tau$ is a type, there exists another type $\abs{\nu}{\tau}$ called
the \emph{abstraction} of $\tau$ by $\nu$.  

\begin{figure}[tb]
  \[
  \begin{array}{rc}
    \multicolumn{2}{l}{\text{Swapping}}\\
    (CS_1)& \forall a{:}\nu, x{:}\tau.\; (a~a)\act x \eq x\\
    (CS_2)& \forall a,a'{:}\nu,x{:}\tau.\; (a~a')\act(a~a')\act x \eq x\\
    (CS_3)& \forall a,a'{:}\nu.\; (a~a')\act a \eq a'\\
    \multicolumn{2}{l}{\text{Equivariance}}\\
    (CE_1)& \forall a,a'{:}\nu,b,b'{:}\nu',x{:}\tau.\; (a~a')\act (b~b')\act x \eq ((a~a')\act b~(a~a')\act b')\act(a~a')\act x\\
    (CE_2)& \forall a,a'{:}\nu,b{:}\nu',x{:}\tau.\; b\fresh x \impp (a~a')\act b \fresh (a~a') \act x\\
    (CE_3)& \forall a,a'{:}\nu,\vec{x}{:}\vec{\tau}.\; (a~a')\act f(\bar{x}) \eq f((a~a')\act \vec{x}) \\
    (CE_4)& \forall a,a'{:}\nu,\vec{x}{:}\vec{\tau}.\; p(\vec{x}) \impp p((a~a')\act \vec{x}) \\
    (CE_5) & \forall b,b'{:}\nu', a{:}\nu, x{:}\tau.\; (b~b')\act (\abs{a}x) \eq \abs{(b~b')\act a} ((b~b')\act x)\\
    \multicolumn{2}{l}{\text{Freshness}} \\
    (CF_1) & \forall a,a'{:}\nu, x{:}\tau.\; a\fresh x \andd a' \fresh x \impp (a~a')\act x \eq x\\
    (CF_2) & \forall a,a'{:}\nu.\;a \fresh a' \iff a \not\eq a'\\
    (CF_3) & \forall a{:}\nu,a'{:}\nu'.\; a\fresh a'\\
    (CF_4) & \forall \vec{x}{:}\vec{\tau}.\; \exists a{:}\nu.~ a \fresh \vec{x}\\
    \multicolumn{2}{l}{\text{$\new$-quantifier}} \\
    (CQ)& \forall \vec{x}. (\new a{:}\nu.\;\phi) \iff (\exists a{:}\nu.\; a\fresh \vec{x} \andd \phi)\\
    \multicolumn{2}{l}{\text{where $FV(\new a.\phi) \subseteq \{\vec{x}\}$}}\\
    \multicolumn{2}{l}{\text{Abstraction}} \\
    (CA_1)& \begin{array}{rcl}
      \forall a,a'{:}\nu, x,x'{:}\tau.\; \abs{a}x \eq \abs{a'}x' &\iff& (a \eq a' \andd x \eq x')\\
      &\orr& (a' \fresh x \andd x' \eq (a~a')\act x)
    \end{array}\\
    (CA_2)& \forall y{:}\abs{\nu}{\tau}.\exists a{:}\nu,x{:}\tau.~ y \eq \abs{a}x
  \end{array}
  \]
\hrule
  \caption{Axioms of Classical Nominal Logic}\labelFig{nl-ax}
\end{figure}

Besides possessing equality at every type, nominal logic includes a
binary \emph{freshness} relation symbol $fresh_{\nu\tau} : \nu,
\tau \to o$ for each name type $\nu$ and type $\tau$. In addition, nominal
logic includes two special function symbols $swap_{\nu\tau} : \nu
, \nu , \tau \to \tau$ and $abs_{\nu\tau}: \nu , \tau
\to \abs{\nu}{\tau}$, called \emph{swapping} and \emph{abstraction}
respectively.  When there is no risk of confusion, we abbreviate
formulas of the form $fresh_{\nu\tau}(a,t)$ as $a \fresh t$, and terms
of the form $swap_{\nu\tau}(a,b,t)$ and $abs_{\nu\tau}(a,t)$ as
$\swap{a}{b}{t}$ and $\abs{a}{t}$ respectively.
In addition, besides the ordinary $\forall$ and $\exists$ quantifiers,
nominal logic possesses a third quantifier, called the
\emph{fresh-name quantifier} and written $\new$.  A $\new$-quantified
formula $\new x{:}\nu.\phi$ may be constructed for any name-type
$\nu$.

Pitts presented a Hilbert-style axiom system for nominal logic shown
in \refFig{nl-ax}.  The axioms are divided into five groups:
\begin{itemize}
\item \emph{Swapping axioms ($CS$)}: describe the behavior of the
  swapping operation: swapping a name for itself has no effect ($CS_1$),
  swapping is involutive ($CS_2$), and swapping exchanges names ($CS_3$).

\item \emph{Equivariance axioms ($CE$)}: prescribe the
  \emph{equivariance} property, namely that all relations are
  preserved by and all function symbols commute with swapping.  In
  particular, ($CE_1$) says that the swapping function symbol itself is
  equivariant; ($CE_2$) says that freshness is equivariant, ($CE_3$) says that all
  other function symbols are equivariant, and ($CE_4$) says that all other
  relation symbols are equivariant.  Also, ($CE_5$) says that abstraction is
  equivariant.

\item \emph{Freshness axioms ($CF$)}: describe the behavior of the
  freshness relation (and its interaction with swapping).  $(CF_1)$
  says that two names fresh for a value can be exchanged without
  affecting the value.  ($CF_2$) says that freshness coincides with
  inequality for names.  ($CF_3$) says that distinct name-types are
  disjoint.  Finally, ($CF_4$) expresses the \emph{freshness
    principle}, namely, that for any finite collection of values, a
  name fresh for all the values simultaneously may be chosen.

\item \emph{$\new$-quantifier axiom scheme $(CQ)$}: Pitts' original
  formalization introduced no new inference rules for $\new$.
  Instead, $\new$ was defined using the axiom scheme $Q$, which
  asserts $\forall\vec{x}.(\new a.\phi \iff \exists a.a \fresh \bar{x}
  \andd \phi)$, where $FV(\phi) \subseteq \{a,\vec{x}\}$.

\item \emph{Abstraction axioms $(CA)$}: These define special
  properties of the abstraction function symbol.  Specifically,
  $(CA_1)$ defines equality on abstractions as either structural
  equality or equality up to ``safe'' renaming of bound names.  Gabbay
  and Pitts showed that this generalizes $\alpha$-equivalence in,
  for example, the lambda-calculus~\cite{gabbay02fac}; we shall not
  repeat the argument here.  Axiom $(CA_2)$ states a surjectivity
  property for abstraction: any value of abstraction type
  $\abs{\nu}{\tau}$ can be written as $\abs{a}{x}$ for some name
  $a:\nu$ and value $x:\tau$.
\end{itemize}

\subsection{Gentzen systems}\labelSec{the-problem}

While admirable from a reductionist point of view, Hilbert systems
have well-known deficiencies: Hilbert-style proofs can be highly
nonintuitive and circuitous.  Instead, Gentzen-style \emph{natural
  deduction} and \emph{sequent} systems provide a more intuitive
approach to formal reasoning in which logical connectives are
explained as \emph{proof-search} operations.  Gentzen systems are
especially useful for computational applications, such as automated
deduction and logic programming.  Such systems are also convenient for
relating logics by proof-theoretic translations.

Gentzen-style rules for $\new$ have been considered in previous work.
Pitts~\cite{pitts03ic} proposed sequent and natural deduction rules
for $\new$ based on the observation that
\[ \forall a{:}\nu. (a \fresh
\bar{x} \impp \phi(a,\bar{x}))\impp \new a{:}\nu.\phi(a,\bar{x}) \impp \exists a{:}\nu. (a \fresh \bar{x} \andd \phi(a,\bar{x})) \;.\]
These rules (see \refFig{evolution-rules}(NL)) are symmetric,
emphasizing $\new$'s self-duality.  However, they are not closed under
substitution, which complicates proofs of  cut-elimination or
proof-normalization properties.

Gabbay~\cite{gabbay07jal} introduced an intuitionistic natural
deduction calculus called Fresh Logic (\FL) and studied semantic
issues including soundness and completeness as well proving
proof-normalization. Gabbay and Cheney~\cite{gabbay04lics}
presented a similar sequent calculus called \FLseq.  
Both \FL and \FLseq had complex rules for $\new$.  In \FL, Gabbay
introduced a technical device called \emph{slices} for obtaining rules
that are closed under substitution.  (For the purpose of this
discussion, it is not necessary to go into the details of what slices
are, since we will show that we can do without them.)
Technically, a slice $\phi[a \#\bar{u}]$
of a formula $\phi$ is a decomposition of the formula as
$\phi(a,\bar{x})[\bar{u}/\bar{x}]$ for fresh variables $\bar{x}$, such
that $a$ does not appear in any of the $\bar{u}$.  Slices were used in
both \FL and \FLseq to deal with $\new$ (see
\refFig{evolution-rules}(\FL,\FLseq)).  The slice-based rules shown in
\refFig{evolution-rules}(\FLseq) \emph{are} closed under substitution,
so proving cut-elimination for these rules is relatively
straightforward once several technical lemmas involving slices have
been proved.  Noting that the \FLseq rules are structurally similar to
$\allL$ and $\exR$, respectively, Gabbay and Cheney observed that
alternate rules in which $\newL$ was similar to $\exL$ and $\newR$
similar to $\allR$ were possible (see
\refFig{evolution-rules}($FL_{Seq}'$)).  These rules seem simpler and
more deterministic; however, they still involve slices.

Experience gained in the process of implementing \aprolog, a logic
programming language based on nominal logic~\cite{cheney04iclp},
suggests a much simpler reading of $\new$ as a proof-search operation
than that implied by the $\FL$-style rules.  In \aprolog, when a
$\new$-quantifier is encountered (either in a goal or program clause),
proof search proceeds by generating a fresh name $\Aa$ to be used for
the $\new$-quantified name.  Besides satisfying a syntactic freshness
requirement (like
eigenvariables in $\forall$-introduction or $\exists$-elimination
rules), the fresh name is also required to be \emph{semantically
  fresh}, that is, fresh for all values appearing in the derivation up
to the point at which it is generated.  In contrast, the proof-search
interpretation suggested by $\FL$-style rules is to search for a
suitable slice of the $\new$-quantified formula.  This reading seems
much less deterministic than that employed in \aprolog.  

\begin{figure}[tb]
\[
\begin{array}{ccl}
\infer[\newR]{\Gamma \seq \new a{:}\nu.\phi,\Delta}{\Gamma,a \fresh \bar{x} \seq \phi,\Delta& (\bar{x} = FV(\Gamma,\new a.\phi,\Delta))}
& 
\infer[\newL]{\Gamma,\new a{:}\nu.\phi\seq \Delta }{\Gamma,a \fresh \bar{x},\phi \seq \Delta& (\bar{x} = FV(\Gamma,\new a{:}\nu.\phi,\Delta))}
& (NL)\smallskip\\
\infer[\new{I}]{\Gamma \nd \new a{:}\nu. \phi[a\#\bar{t}]}{\Gamma \nd u \fresh \bar{t} & \Gamma \nd \phi[a\#\bar{t}][u/a]}
&
\infer[\new{E}]{\Gamma\seq \psi }{\Gamma \nd \new \Aa{:}\nu.\phi[a\#\bar{t}] & \Gamma \nd u \fresh \bar{t} & \Gamma,\phi[u/a] \nd \psi}
& (FL)\smallskip\\
\infer[\newR]{\Gamma,u \fresh \bar{t} \seq \new a{:}\nu.\phi[a\#\bar{t}]}{\Gamma,u \fresh \bar{t} \seq \phi[u/a]}\quad
&
\infer[\newL]{\Gamma,u \fresh \bar{t},\new a{:}\nu. \phi[a\#\bar{t}]\seq \psi }{\Gamma,u \fresh \bar{t},\phi[u/a] \seq \psi}
& (\FLseq)\smallskip\\
\infer[\newR]{\Gamma \seq \new a{:}\nu. \phi[a\#\bar{t}]}{\Gamma,a \fresh \bar{t} \seq \phi & (a \not\in FV(\Gamma,\psi))}&
\infer[\newL]{\Gamma,\new a{:}\nu. \phi[a\#\bar{t}]\seq \psi }{\Gamma,a \fresh \bar{t},\phi \seq \psi&(a \not\in FV(\Gamma,\psi))}
& (FL_{Seq}')
\smallskip\\
\infer[\newR]{\Sigma;\Gamma \seq \new \Aa{:}\nu.\phi}
        {\Sigma\#\Aa{:}\nu;\Gamma \seq \phi & (\Aa \notin \Sigma)}
&
\infer[\newL]{\Sigma;\Gamma, \new \Aa{:}\nu.\phi \seq \psi}
        { \Sigma\#\Aa{:}\nu;\Gamma, \phi \seq \psi & (\Aa \notin \Sigma)}
&(\NLseq)
\end{array}\]
\hrule
\caption{Evolution of rules for $\new$}\labelFig{evolution-rules}
\end{figure}

In this article we present a simplified sequent calculus for nominal
logic, called \NLseq, in which slices are not needed in the rules for
$\new$ (or anywhere else).  Following Urban, Pitts, and
Gabbay~\cite{urban04tcs,gabbay07jal}, and our prior
work~\cite{cheney06jsl}, we employ a new syntactic class of
\emph{name-symbols} $\Aa,\Ab,\ldots$ different from ordinary variables
$x,y,z,\ldots$.  Like variables, such name-symbols may be bound (by
$\new$), but unlike variables, two distinct name-symbols always denote
distinct name values.  As explained in our previous
paper~\cite{cheney06jsl}, name-symbols can be used to construct ground
terms, which is convenient form the perspective of studying Herbrand
models and consistency.  In place of slices, we introduce contexts
that encode information about freshness as well as identifying the
types of variables and name-symbols.  Specifically, contexts
$\Sigma\#\Aa{:}\nu$ may be formed by adjoining a \emph{fresh
  name-symbol} $\Aa$ which is also assumed to be semantically fresh
for any value mentioned in $\Sigma$.  Our rules for $\new$
(\refFig{evolution-rules}(\NLseq)) are in the spirit of the original
rules and are very simple.

Besides presenting the sequent calculus and proving structural
properties such as cut-elimination, we verify that \NLseq and Pitts'
axiomatization \NL are equivalent.  We also present a syntactic proof
of the consistency of the nonlogical rules, which together with
cut-elimination implies consistency of the whole system.

The structure of this article is as follows: \refSec{sequent-calculus}
presents the sequent calculus \NLseq along with proofs of structural
properties.  \refSec{applications} discusses several applications,
including proofs of consistency and equivalence of \NLseq
to \NL.  \refSec{concl} concludes.  

This article builds upon prior work by Gabbay and
Cheney~\cite{gabbay04lics} and Gabbay~\cite{gabbay07jal}, which
introduced sequent and natural-deduction calculi for nominal logic,
based on slices.  The closest-related prior publication is
Cheney~\cite{cheney05fossacs}, which introduced a single-conclusion,
intuitionistic version of $\NLseq$ with the simpler rules for
$\new$-quantifiers shown above.  This article generalizes the approach taken there
and provides detailed proofs of the main results, along with proofs of
new results including equivalence to classical nominal
logic.

\section{Sequent Calculus}\labelSec{sequent-calculus}
\subsection{Syntax}

The types $\tau$, terms $t$, and formulas $\phi$ of \NLseq are
generated by the following grammar:
\[\begin{array}{rcll}
  \tau,\sigma &::=& \delta \mid \nu 
\mid \abs{\nu}{\tau}\\
  t,u &::=& x\mid \Aa 
\mid c  \mid f(\vec{t})
& \|~ \swap{a}{b}{t} \mid \abs{a}{t}\\
  \phi,\psi &::=& \true \mid \false \mid p(\vec{t}) \mid \phi \andd
  \psi \mid \phi \orr \psi \mid \phi \impp \psi \mid \forall
  x{:}\tau.\phi \mid \exists x{:}\tau.\phi \mid \new \Aa{:}\nu.\phi
  &\|~  t \eq u \mid
  t \fresh u
\end{array}\]
The constructs to the right of $\|$ are syntactic sugar that are
definable in terms of the core language as explained below; we list
them in the grammar for ease of reference.
The base types are datatypes $\delta$ and name-types $\nu$; additional
types are formed using the abstraction type constructor.  Terms are
first-order, with variables $x,y$ are drawn from a countably infinite
set $\Var$; also, name-symbols $\Aa,\Ab$ are drawn from a countably
infinite set $\Name$ disjoint form $\Var$.  The letters $a,b$ are typically
used for terms of some name-type
$\nu$.  
Negation and
logical equivalence are defined as follows:
\[\nott \phi = (\phi \impp \bot)\qquad \phi \iff \psi = (\phi \impp
\psi) \andd (\psi \impp \phi)\]

We assume given a signature that maps constant symbols $c$ to types
$\delta$, function symbols $f$ to sorts $\tau_1,\ldots,\tau_n \to
\delta$, and relation symbols to sorts $\tau_1,\ldots,\tau_n \to o$, and containing at least the following declarations:
\[\begin{array}{rclcrcl}
  swap_{\nu\tau} &:& \nu, \nu , \tau \to \tau &&
  abs_{\nu\tau} &:& \nu , \tau \to \abs{\nu}{\tau}\\
eq_{\tau} &:& \tau,\tau \to o && fresh_{\nu,\tau} &:& \nu,\tau \to o
\end{array}\]
for name-types $\nu$ and types $\tau$.  The subscripts are dropped
when clear from context. The notations $\swap{a}{b}{t}$ and
$\abs{t}{u}$ are syntactic sugar for the terms $swap(a,b,t)$ and
$abs(t,u)$, respectively.  Likewise, $t \eq u$ and $t \fresh u$ are
syntactic sugar for $eq(t,u)$ and $fresh(t,u)$,
respectively.  The functions
$FV(\cdot)$, $FN(\cdot)$, $FVN(\cdot)$ calculate the sets of free
variables, name-symbols, or both variables and name-symbols of a term
or formula (see \refFig{free-vars-names}).  We lift the swapping
operation to formulas as shown in \refFig{formula-swapping}.

\begin{figure}[tb]
\[\begin{array}{rcll}
  FV(x) &=& \{x\}\\
  FV(\Aa) &=& \emptyset\\
  FV(Qx{:}\sigma.\phi) &=& FV(\phi) - \{x\} & (Q \in \{\forall, \exists\})\\
  FV(\new \Aa{:}\nu.\phi) &=& FV(\phi)\\
\\
  FN(x) &=& \emptyset\\
  FN(\Aa) &=& \{\Aa\}\\
  FN(Qx{:}\sigma.\phi) &=& FN(\phi) & (Q \in \{\forall, \exists\})\\
  FN(\new\Aa{:}\nu.\phi) &=& FN(\phi) - \{\Aa\}\\
  \\
  F\alpha(c) = F\alpha(\top) = F\alpha(\bot) &=& \emptyset\\
  F\alpha(f(\vec{t})) = F\alpha(p(\vec{t})) &=& \bigcup F\alpha(t_i)\\
 F\alpha(\phi \circ \psi) &=& F\alpha(\phi) \cup F\alpha(\psi) & (\circ \in \{\andd, \orr, \impp\})\\
  F\alpha(\swap{a}{b}{t}) &=& F\alpha(a) \cup F\alpha(b) \cup F\alpha(t)\\
  F\alpha(\abs{a}{t}) &=& F\alpha(a) \cup F\alpha(t)\\
  \\
  FVN(t) &=& FV(t) \cup FN(t)
\end{array}\]
\hrule
\caption{Free variables and names (note $F\alpha$ stands for either
  $FV$ or $FN$)}\labelFig{free-vars-names}
\[\begin{array}{rcll}
\swap{a}{b}{\phi} &=& \phi & (\phi \in \{\top,\bot\})\\
\swap{a}{b}{p(\vec{t})} &=& p(\swap{a}{b}{\vec{t}})\\
\swap{a}{b}{\phi \circ \psi} &=& \swap{a}{b}{\phi} \circ
\swap{a}{b}{\psi} & (\circ \in \{\andd,\orr,\impp\})\\
\swap{a}{b}{Qx{:}\sigma.\phi} &=& Qx{:}\sigma.\swap{a}{b}{\phi} &
(Q \in \{\forall,\exists\}, x \notin FV(a)\cup FV(b))\\
\swap{a}{b}{\new \Aa{:}\nu.\phi} &=& \new \Aa{:}\nu.\swap{a}{b}{\phi}
& (\Aa \notin FN(a) \cup FN(b))
\end{array}
\]
\hrule
\caption{Swapping for formulas}\labelFig{formula-swapping}
\end{figure}

The \emph{typing contexts} used in \NLseq are generated by the
grammar:
\[\Sigma ::= \cdot \mid \Sigma,x{:}\tau \mid \Sigma\# \Aa{:}\nu\]
We often write  $\cdot,x{:}\tau$ and $\cdot\#\Aa{:}\nu$ to
$x{:}\tau$ and $\Aa{:}\nu$ respectively. We write $\omega$ for a term that may be
either a name-symbol $\Aa$ or a variable $x$.  The $\Sigma\#\Aa{:}\nu$
binding indicates that $\Aa$ is a name of type $\nu$ and is assumed to
be fresh with respect to all names and variables in $\Sigma$.
We write $\omega {:}\tau \in
\Sigma$ if the binding $\omega{:}\tau$ is present in $\Sigma$.
We write $\Sigma,\Sigma'$ for the result of concatenating two contexts
such that $FVN(\Sigma) \cap FVN(\Sigma') = \emptyset$.

We write $\Sigma \nd t : \tau$ or $\Sigma \nd \phi : \propp$ to indicate
that $t$ is a well-formed term of type $\tau$ or $\phi$ is a
well-formed formula.  From the point of view of typechecking, the
additional freshness information in the context is irrelevant.  The
rules for typechecking 
(shown in \refFig{well-formedness}) are standard, except for the rules
for freshness and the $\new$-quantifier.  Quantification using
$\forall$ and $\exists$ is only allowed over types not mentioning $o$;
$\new$-quantification is only allowed over name-types.

\begin{figure}[tb]
\begin{mathpar}
\infer{\Sigma\nd c : \delta}{c : \delta}
\and
\infer{\Sigma \nd f(\vec{t}) : \delta}{ f:\tau_1,\ldots,\tau_n \to
  \delta & \Sigma \nd t_i : \tau_i}
%
%
\and
\infer{\Sigma\nd \omega:\tau}{\omega:\tau \in
  \Sigma}
%
%
\and
\infer{\Sigma\nd \top : \propp}{} 
\\
\infer{\Sigma\nd \bot : \propp}{} 
\and
\infer{\Sigma \nd \phi \circ \psi : \propp}{\Sigma\nd \phi, \psi : \propp & (\circ \in \{\andd,\orr,\impp\})}
\and
\infer{\Sigma \nd a \fresh t : \propp}{\Sigma \nd a : \nu & \Sigma \nd t : \tau}
\and
\infer{\Sigma \nd t \eq u : \propp}{\Sigma \nd t,u : \tau}
\and
\infer{\Sigma\nd \forall x{:}\tau.\phi: \propp}{\Sigma,x{:}\tau \nd \phi : \propp}
\and
\infer{\Sigma\nd \exists x{:}\tau.\phi : \propp}{\Sigma,x{:}\tau \nd \phi : \propp}
\and
\infer{\Sigma\nd \new \Aa{:}\nu .\phi : \propp}{\Sigma\#\Aa{:}\nu\nd \phi : \propp}
\end{mathpar}
\hrule
\caption{Well-formedness rules}\labelFig{well-formedness}
\end{figure}

\begin{definition}
Let $Tm_\Sigma = \{t \mid \Sigma \nd t : \tau\}$ be the set of
    well-formed terms in context $\Sigma$. 
  \begin{itemize}
  \item We associate a set of
    freshness formulas $|\Sigma|$ to each context $\Sigma$ as follows:
    \[|\cdot| = \emptyset \qquad |\Sigma,x:\tau| = |\Sigma|\qquad
    |\Sigma\#\Aa:\nu| = |\Sigma| \cup \{\Aa \fresh t \mid t \in
    Tm_\Sigma\}\]
    For example, $\Aa \fresh x$, $\Ab \fresh \Aa$ and $\Ab \fresh
    f(x,y)$ are in $|x{:}\tau\#\Aa{:}\nu,y{:}\sigma\#\Ab{:}\nu'|$.  
\item We
    say that $\Sigma'$ is stronger than $\Sigma$ ($\Sigma \leq
    \Sigma'$) if $Tm_{\Sigma} \subseteq Tm_{\Sigma'}$ and $|\Sigma|
    \subseteq |\Sigma'|$.  For example, $\Aa{:}\nu,x{:}\tau \leq
  x{:}\tau\#\Aa{:}\nu,y{:}\sigma$.
\item We say that $\Aa{:}\nu \in \Sigma$ if $\Sigma =
  \Sigma'\#\Aa{:}\nu,\Sigma''$ for some contexts $\Sigma',\Sigma''$
  and similarly $x{:}\tau \in \Sigma$ means that $\Sigma =
  \Sigma',x{:}\tau,\Sigma''$ for some contexts $\Sigma',\Sigma''$.
\item We say that $\Aa$ is fresh for $\Sigma$ if $\Aa$ is not among the names 
 appearing in $\Sigma$; we write $\Aa \notin\Sigma$ to indicate that
 this is the case. Similarly, we write $x \notin \Sigma$ to indicate
 that variable $x$ does not appear in $\Sigma$.
\end{itemize}
\end{definition}
The following routine properties hold:
\begin{lemma}[Term Weakening]
If $\Sigma\nd t : \tau$ and $\Sigma \leq \Sigma'$ then $\Sigma' \nd t : \tau$.
\end{lemma}
\begin{lemma}[Term Substitution]
  If $\Sigma \vdash t : \tau$ and $\Sigma,x{:}\tau,\Sigma' \vdash u :
  \tau'$ then $\Sigma,\Sigma' \vdash u[t/x] : \tau'$.
\end{lemma}

\subsection{The Rules}
\label{sec:rules}

Judgments are of the form $\Sigma;\Gamma \seq \Delta$, where $\Sigma$
is a typing context and $\Gamma,\Delta$ are multisets of formulas.  We
define classical and intuitionistic versions of \NLseq.
\emph{Classical \NLseq} is based on the classical sequent calculus
$\mathbf{G3c}$ (see \refFig{seqrules-fole}).  The new rules defining
\NLseq are defined in Figures~\ref{fig:nonlogical-axioms} and
\ref{fig:nonlogical-rules}.  \NLseq includes two additional
\emph{logical rules}, $\newL$ and $\newR$, as already shown in
\refFig{evolution-rules}.  In addition, \NLseq includes several new
\emph{nonlogical rules} defining the properties of swapping, equality,
freshness and abstraction.  (The standard rules involving equality in
\refFig{seqrules-fole} are also considered nonlogical rules.)

Many of the nonlogical rules correspond to first-order universal
axioms of nominal logic (\refFig{nonlogical-axioms}), which may be
incorporated into sequent systems in a uniform fashion using the $Ax$
rule schema without affecting
cut-elimination~\cite{negri01structural}.  Here, we write an axiom of the
form $P_1 \andd \cdots \andd P_n \impp Q_1 \vee \cdots \vee Q_m$ as
$\Andd\vec{P} \impp \Orr \vec{Q}$.  To illustrate, the instances of
this scheme  for axioms $F_3$ and $F_4$ are:
\[
\infer[F_3]{\Sigma;\Gamma,a \fresh a  \seq \Delta}{}
\qquad \infer[F_4]{\Sigma;\Gamma \seq \Delta}{\Sigma;\Gamma,a \fresh b \seq
  \Delta & \Sigma;\Gamma,a \eq b \seq \Delta}
\]
The key point of this treatment of nonlogical rules is that they act
only on the hypothesis set $\Gamma$, so they do not  introduce new
principal cut cases in the proof of cut-elimination.

The remaining nonlogical rules are as follows.  Rule $A_2$ expresses
an invertibility property for abstractions: two abstractions are equal
only if they are structurally equal or equal by virtue of $A_1$.
$A_3$ says that all values of abstraction type are formed using the
abstraction function symbol.  The $F$ rule expresses the freshness
principle: that a name fresh for a given context may always be chosen.
It is important to note that the fresh name chosen in $F$ may be of
any name type $\nu$, and thus, all name types are inhabited; however,
base data types $\delta$ could be empty, and an abstraction type
$\abs{\nu}{\tau}$ is inhabited if and only if $\tau$ is.  Finally, the
$\Sigma\#$ rule allows freshness information to be extracted from the
context $\Sigma$.  It states that in context $\Sigma$, any constraint
in $|\Sigma|$ is valid.

\begin{remark}
  Although we have motivated some choices in \NLseq in terms of
  proof-search behavior based on experience with \aprolog, some rules,
  such as $A_3$ and $\Sigma\#$, do not have particularly pleasant
  proof-search properties. It is fair to say that \NLseq addresses
  only the proof search complexity arising from the $\new$-quantifier
  and (to some extent) freshness but does not help very much with the
  complexity arising from equational/freshness reasoning.  In
  \aprolog, special cases of these problems are dealt with using
  nominal unification and freshness constraint solving; in this paper
  we aim to deal with full nominal logic.
\end{remark}

The naming of the nonlogical rule groups corresponds to that used by
Pitts: the axioms are divided into groups for swapping $(S)$,
equivariance $(E)$, freshness $(F)$, and abstraction $(A)$.  The $(Q)$
axiom is replaced by the logical rules $\newL$ and $\newR$.

\begin{figure}[tb]
\[
\begin{array}{cc}
\hyp[hyp]{\Sigma ;\Gamma,p(\vec{t}) \seq p(\vec{t}),\Delta}
\smallskip\\
\infer[\trueR]{\Sigma;\Gamma \seq \top,\Delta}{} 
& 
\infer[\falseL]{\Sigma;\Gamma,\bot \seq \Delta}{}\smallskip\\
\infer[\andR]{\Sigma;\Gamma \seq \phi \andd \psi,\Delta}
{\Sigma;\Gamma \seq \phi,\Delta & \Sigma;\Gamma \seq \psi,\Delta} 
& 
\infer[\andL]{\Sigma;\Gamma, \phi_1 \andd \phi_2 \seq \Delta}
{\Sigma;\Gamma, \phi_1 , \phi_2 \seq \Delta}\smallskip\\
\infer[\orR]{\Sigma;\Gamma \seq \phi_1 \orr \phi_2,\Delta}{\Sigma;\Gamma \seq \phi_1,\phi_2,\Delta}
& 
\infer[\orL]{\Sigma;\Gamma,\phi \orr \psi \seq \Delta}{\Sigma;\Gamma ,\phi\seq \Delta & \Gamma, \psi \seq \Delta} \smallskip\\
\infer[\impR]{\Sigma;\Gamma \seq \phi \impp \psi,\Delta}{\Sigma;\Gamma,\phi \seq\psi,\Delta}
& 
\infer[\impL]{\Sigma;\Gamma,\phi \impp \psi \seq \Delta}
      {\Sigma;\Gamma \seq \phi,\Delta & 
        \Sigma;\Gamma, \psi \seq \Delta}\smallskip\\
\infer[\allR]{\Sigma;\Gamma \seq \forall x{:}\sigma.\phi,\Delta}
        {\Sigma,x{:}\sigma ;\Gamma \seq \phi,\Delta & 
        (x \not\in \Sigma)}
&
\infer[\allL]{\Sigma;\Gamma, \forall x{:}\sigma.\phi\seq \Delta}
        {\Sigma \vdash t : \sigma & 
        \Sigma;\Gamma, \forall x{:}\sigma.\phi,\phi\subs{t}{x} \seq \Delta}
\smallskip\\
\infer[\exR]{\Sigma;\Gamma \seq \exists x{:}\sigma.\phi,\Delta}
        {\Sigma \vdash t : \sigma
&       \Sigma;\Gamma \seq \exists x{:}\sigma.\phi,\phi\subs{t}{x},\Delta}
&
\infer[\exL]{\Sigma;\Gamma,\exists x{:}\sigma.\phi \seq \Delta}
        {\Sigma,x{:}\sigma;\Gamma, \phi\seq \Delta & 
          (x \not\in \Sigma)}\smallskip\\
  \infer[{\eq}R]{\Sigma;\Gamma \seq \Delta}{\Sigma;\Gamma, t \eq t \seq \Delta}
  &
  \infer[{\eq}S]{\Sigma;\Gamma, t \eq u, P(t) \seq\Delta}
  {\Sigma;\Gamma,t\eq u,P(t),P(u) \seq \Delta}

\end{array}\]
\hrule
\caption{Classical first-order equational sequent calculus
  ($\mathbf{G3c}$)}\labelFig{seqrules-fole}
\end{figure}
\begin{figure}[tb]
\[
\begin{array}{lc}
(S_1) & \swap{a}{a}{x} \eq x  
\\
(S_2) & \swap{a}{b}{\swap{a}{b}{x}} \eq x
\\
(S_3) & \swap{a}{b}{a} \eq b 
 \\
(E_1) & \swap{a}{b}{c} \eq c
\\
(E_2) & \swap{a}{b}{f(\vec{t})} \eq f(\swap{a}{b}{\vec{t}})
\\
(E_3) & p(\vec{t}) \impp p(\swap{a}{b}{\vec{t}})
\end{array}
\begin{array}{lc}
(F_1) & a \fresh x \andd b \fresh x \impp \swap{a}{b}{x} \eq x
\\
(F_2) & a \fresh b \quad (a:{\nu},b:{\nu'},\nu \not\equiv \nu')
\\
(F_3) & a \fresh a \impp \false
  \\
(F_4) & a \fresh b \orr a \eq b
\\
(A_1) & a \fresh y \andd x \eq \swap{a}{b}{y} \impp \abs{a}{x} \eq \abs{b}{y}
\end{array}\]
\hrule
\caption{Equational and freshness axioms}\labelFig{nonlogical-axioms}
\[\begin{array}{c}
  \infer[Ax\qquad \text{$\Andd \bar{P} \impp \Orr \bar{Q}$ an axiom
    instance in \refFig{nonlogical-axioms}}]{\Sigma;\Gamma, \Andd \bar{P}\seq \Delta}
  {\Sigma;\Gamma, \Andd \bar{P},Q_1 \seq \Delta  &\cdots&  \Sigma;\Gamma, \Andd \bar{P},Q_n \seq
    \Delta}
  \smallskip\\
  \infer[A_2]{\Sigma;\Gamma,\abs{a}{t} \eq \abs{b}{u} \seq \Delta}{\Sigma;\Gamma,\abs{a}{t} \eq \abs{b}{u},a \eq b,t\eq u \seq \Delta & \Sigma;\Gamma,\abs{a}{t} \eq \abs{b}{u}, a \fresh u, t \eq\swap{a}{b}{u} \seq \Delta}\smallskip\\
  \infer[A_3]{\Sigma;\Gamma \seq \Delta}
  {\Sigma \vdash t: \abs{\nu}{\sigma} & \Sigma,a{:}\nu,x{:}\sigma;
    \Gamma,t \eq \abs{a}{x} \seq \Delta& (a,x \notin \Sigma)}\smallskip\\
  \infer[F]{\Sigma;\Gamma \seq \Delta}
  {\Sigma\#\Aa{:}\nu;\Gamma \seq  \Delta 
    & (\Aa \notin \Sigma) }
  \qquad
  \infer[\Sigma\#]{\Sigma;\Gamma \seq \Delta}
  { \Sigma;\Gamma,\Aa\fresh t \seq \Delta&(\Aa \fresh t \in |\Sigma|) }\smallskip\\
\infer[\newR]{\Sigma;\Gamma \seq \new \Aa{:}\nu.\phi,\Delta}
        {\Sigma\#\Aa{:}\nu;\Gamma \seq \phi,\Delta & (\Aa \notin \Sigma)}
\qquad
\infer[\newL]{\Sigma;\Gamma, \new \Aa{:}\nu.\phi \seq \Delta}
        { \Sigma\#\Aa{:}\nu;\Gamma, \phi \seq \Delta & (\Aa \notin \Sigma)}
\end{array}\]
\hrule
\caption{Nonlogical and $\new$-quantifier
  rules}\labelFig{nonlogical-rules}
\end{figure}
\begin{figure}[tb]
\[\begin{array}{ccc}
\infer[W]{\Sigma;\Gamma,\phi \seq \Delta}{\Sigma;\Gamma \seq \Delta}
&
\hyp[hyp^*]{\Sigma;\Gamma,\phi \seq \phi,\Delta}
&
\infer[cut]{\Sigma;\Gamma,\Gamma' \seq \Delta,\Delta'}
        {\Sigma;\Gamma \seq \phi,\Delta & \Sigma;\Gamma',\phi \seq \Delta'}
\medskip\\
\infer[C]{\Sigma;\Gamma,\phi\seq \Delta}{\Sigma;\Gamma,\phi,\phi \seq \Delta}
&
\infer[EVL]{\Sigma;\Gamma,\phi \seq \Delta}{\Sigma;\Gamma,\swap{a}{b}\phi \seq \Delta}
&
\infer[EVR]{\Sigma;\Gamma \seq \Delta,\phi}{\Sigma;\Gamma \seq \swap{a}{b}\phi,\Delta}
\end{array}\]
\hrule
\caption{Some admissible rules of \NLseq}\labelFig{admissible-rules}
\end{figure}

\subsection{Structural Properties}

\refFig{admissible-rules} lists some additional rules, including weakening,
contraction, general form of hypothesis and equivariance rules, and
cut.  We will now prove
their admissibility.  Note that these rules are not part of the
definition of \NLseq, and so in proving admissibility, it suffices to
consider only derivations using the core rules introduced in Section~\ref{sec:rules}.

We now list some routinely-verified syntactic properties of \NLseq.
We write $\nd_n J$ to indicate that judgment $J$ has a derivation of
height at most $n$.

\begin{lemma}[Weakening]
  If $\nd_n\Sigma;\Gamma \seq \Delta$ is derivable then so is
  $\nd_n\Sigma;\Gamma,\phi \seq \Delta$.  Similarly,
  $\nd_n\Sigma;\Gamma \seq \Delta,\phi$.
\end{lemma}
\begin{lemma}[Context Weakening]
  If $\nd_n\Sigma;\Gamma \seq \Delta$ and $\Sigma \leq \Sigma'$ then 
  $\nd_n\Sigma' ; \Gamma \seq \Delta$.
\end{lemma}
\begin{lemma}[Substitution]
  If $\nd_n\Sigma \vdash t : \tau$ and $\Sigma,x{:}\tau,\Sigma';\Gamma
  \seq \Delta$ then $\nd_n\Sigma,\Sigma' ; \Gamma[t/x] \seq \Delta[t/x]$.
\end{lemma}
\begin{proof}
  The interesting cases are for the new rules, specifically,
  nonlogical rules, $\newL$, and $\newR$.  All of the nonlogical rules
  are closed under substitution; in particular, for $\Sigma\#$ we have
  $\Aa \fresh u \in |\Sigma,x,\Sigma'|$ then $\Aa \fresh u[t/x] \in
  |\Sigma,\Sigma'|$.  

  For $F$ we have a derivation 
  \[\infer[F]{\Sigma,x{:}\tau,\Sigma';\Gamma\seq \Delta}{\Sigma,x{:}\tau,\Sigma'\#\Aa{:}\nu;\Gamma\seq\Delta}\]
  By induction we have
  $\Sigma,\Sigma'\#\Aa{:}\nu;\Gamma[t/x]\seq\Delta[t/x]$, so we can use $F$
  again to derive $\Sigma,\Sigma';\Gamma[t/x]\seq \Delta[t/x]$.  This
  requires the observation that since $\Aa \notin \Sigma$ and $\Sigma \nd t : \tau$, we must
  have $\Aa\notin FN(t)$.  The proofs for $\newL$ and $\newR$ are
  similar, requiring the additional observation that $(\new
  \Aa{:}\nu.\phi)[t/x] = \new \Aa{:}\nu.(\phi[t/x])$ since $\Aa \not\in FN(t)$.
\end{proof}
The remaining structural transformations do not preserve the height of
derivations.  However, they do preserve the logical
height of the derivation, which is defined as follows.
\begin{definition}
  The \emph{logical height} of a derivation is the maximum number of
  logical rules in any branch of the derivation.  We write $\nd^l_n J$
  to indicate that $J$ has a derivation of logical height $\leq n$.
\end{definition}

Now we consider some structural properties specific to \NLseq.  In the
following, recall the definition of $\swap{a}{b}{\phi}$ given in
\refFig{formula-swapping}.
\begin{lemma}[Admissibility of $EVL$, $EVR$]
  The $EVL$ and $EVR$ rules
\[\infer[EVL]{\Sigma;\Gamma,\phi \seq \Delta}{\Sigma;\Gamma,\swap{a}{b}{\phi} \seq \Delta}\qquad
\infer[EVR]{\Sigma;\Gamma \seq \phi,\Delta}{\Sigma;\Gamma \seq
  \swap{a}{b}{\phi},\Delta}\] where $\phi$ is an arbitrary formula in $Tm_\Sigma$,
are admissible; if the antecedent of $EVL$ or $EVR$ is derivable, then
the respective conclusion has a derivation of the same logical height.
\end{lemma}
\begin{proof}
  We proceed by induction on the lexicographic product of logical
  height and total height to show that if the hypothesis of an
  instance of $EVL$ or $EVR$ has a derivation then the conclusion of
  the respective rule has a derivation of the same logical height.
  
  We first consider $EVL$.  The only interesting cases are when
  $\swap{a}{b}{\phi}$ is principal on the left, otherwise the
  induction step is straightforward.  Furthermore, only the cases for
  $hyp$ and $\impL$ are nontrivial.
  
  If the derivation is of the form
  \[\infer[hyp]{\Sigma;\Gamma,\swap{a}{b}{A} \seq \swap{a}{b}{A},\Delta}{}\]
  then we may derive $\Gamma,A \seq \swap{a}{b}{A},\Delta$ as follows:
  \[\infer[E_3]{\Sigma;\Gamma,A \seq \swap{a}{b}{A},\Delta}
  {\infer[\eq,hyp]{\Sigma;\Gamma,\swap{a}{b}{A} \seq \swap{a}{b}{A},\Delta}{}}\]
  This derivation has the same logical height, 1, as the first.
  
  If the derivation is of the form
  \[\infer[\impL]{\Sigma;\Gamma,\swap{a}{b}{\phi} \impp \swap{a}{b}\psi \seq \Delta}
  {\Sigma;\Gamma, \swap{a}{b}{\phi} \impp \swap{a}{b}{\psi} \seq
    \swap{a}{b}{\phi},\Delta & \Sigma;\Gamma, \swap{a}{b}{\psi} \seq \Delta}
  \]
  then using the admissibility of $EVR$ and $EVL$ on the left and
  $EVR$ on the right (on derivations of smaller logical height) we obtain
  \[\infer[\impL]{\Sigma;\Gamma,{\phi} \impp \psi \seq \Delta}
  {\infer[EVL,EVR]{\Sigma;\Gamma, {\phi} \impp {\psi} \seq {\phi},\Delta}
    {\Sigma;\Gamma,\swap{a}{b}{\phi} \impp \swap{a}{b}{\psi} \seq
      \swap{a}{b}{\phi},\Delta} & \infer[EVL]{\Sigma;\Gamma, {\psi} \seq
      \Delta} {\Sigma;\Gamma, \swap{a}{b}{\psi} \seq \Delta}}\] This
  transformation is obviously logical height-preserving by induction.
  
  For $EVR$, the interesting cases are those for $hyp$ and $\impR$
  where $\swap{a}{b}{\phi}$ is principal on the right.  Suppose the
  derivation is of the form
  \[\infer[hyp]{\Sigma;\Gamma,\swap{a}{b}{A} \seq \swap{a}{b}{A}, \Delta}{}\]
  Then we can derive
  \[\infer[E_3]{\Sigma;\Gamma,\swap{a}{b}{A} \seq A,\Delta}
  {\infer[\eq,hyp]{\Sigma;\Gamma,\swap{a}{b}\swap{a}{b}{A}\seq A,\Delta}{}}\]
  This derivation has the same logical height, 1, as the first.
  
  If the derivation is of the form
  \[\infer[\impR]{\Sigma;\Gamma \seq \swap{a}{b}\phi \impp \swap{a}{b}\psi,\Delta}{\Sigma;\Gamma,\swap{a}{b}\phi \seq \swap{a}{b}\psi,\Delta}\]
  then since $EVL$ and $EVR$ are admissible for all subderivations of
  this derivation, by induction we can derive
  \[\infer[\impR]{\Sigma;\Gamma \seq \phi \impp \psi,\Delta}
  {\infer[EVL,EVR]{\Sigma;\Gamma,\phi \seq \psi,\Delta} {\Sigma;\Gamma,\swap{a}{b}{\phi} \seq
      \swap{a}{b}{\psi},\Delta}}\] This transformation is obviously
  logical height-preserving by induction.
\end{proof}

\begin{lemma}[Swapping Fresh Names]\labelLem{swapping-fresh-names}
Suppose $\Sigma\#\Aa{:}\nu \nd \phi(\Aa) : \propp$ and $\Ab \notin FN(\Sigma\#\Aa{:}\nu)$.  Then the rule
\[\infer{\Sigma\#\Aa{:}\nu \#\Ab{:}\nu ;\Gamma,\phi(\Aa) \seq
  \Delta}{\Sigma\#\Aa{:}\nu \#\Ab{:}\nu ;\Gamma,\phi(\Ab) \seq \Delta}\]
is admissible using nonlogical axioms only.
\end{lemma}
\begin{proof}
  Let $\vec{x} = FV(\Sigma)$.  The derivation is as follows:
\[\infer[\Sigma\#]{\Sigma\#\Aa{:}\nu \#\Ab{:}\nu ;\Gamma,\phi(\Aa) \seq  \Delta}
{\infer[EVL]{\Sigma\#\Aa{:}\nu \#\Ab{:}\nu ;\Gamma,\Aa \fresh \vec{x},\Ab \fresh \vec{x}, \phi(\Aa)\seq  \Delta}
{\infer[Ax]{\Sigma\#\Aa{:}\nu \#\Ab{:}\nu ;\Gamma,\Aa \fresh \vec{x},\Ab \fresh \vec{x}, \swap{\Aa}{\Ab}\phi(\Aa)\seq  \Delta}
{\Sigma\#\Aa{:}\nu \#\Ab{:}\nu ;\Gamma,\Aa \fresh \vec{x},\Ab \fresh \vec{x}, \phi(\Ab)\seq  \Delta}}}\]
where $F_1$ and equational reasoning is used repeatedly to show that
$\swap{\Aa}{\Ab}\phi(\Aa) \impp \phi(\Ab)$.
\end{proof}

\begin{lemma}[Admissibility of $hyp^*$]\labelLem{admissibility-of-hyp}
  The $hyp^*$ rule
\[\infer[hyp^*]{\Sigma;\Gamma,\phi \seq \phi,\Delta}{}\]
where $\phi$ is an arbitrary formula in $Tm_\Sigma$, is admissible.
\end{lemma}
\begin{proof}
  The proof is by induction on the construction of $\phi$.  The cases
  for the ordinary connectives of first-order logic are standard.  The
  case for $\phi = \new \Aa{:}\nu.\phi'$ is as follows.  By induction, we may
  assume that $\Sigma\#\Aa{:}\nu \#\Ab{:}\nu ;\Gamma,\phi(\Ab) \seq
  \phi(\Ab),\Delta$ is derivable.  We derive
  \[\infer[\newL]{\Sigma;\Gamma,\new \Aa{:}\nu .\phi' \seq \new \Aa{:}\nu.\phi',\Delta}{
    \infer[\newR]{\Sigma\#\Aa{:}\nu ;\Gamma,\phi(\Aa) \seq \new \Aa{:}\nu.\phi',\Delta}{
      \infer[\text{\refLem{swapping-fresh-names}}]{\Sigma\#\Aa{:}\nu \#\Ab{:}\nu ;\Gamma, \phi(\Aa) \seq
        \phi(\Ab),\Delta} {\Sigma\#\Aa{:}\nu \#\Ab{:}\nu ;\Gamma,\phi(\Ab) \seq
        \phi(\Ab),\Delta}}}
\]
Using the induction hypothesis, the judgment
$\Sigma\#\Aa{:}\nu \#\Ab{:}\nu ;\Gamma,\phi(\Ab) \seq \phi(\Ab),\Delta$ is derivable,
since it is an instance of $hyp^*$ with a smaller principal formula.
\end{proof}

\begin{lemma}[Inversion]\labelLem{inversion}
  The $\andL$, $\orL$, $\impL$, $\exL$, $\allR$, $\newL$, and $\newR$
  rules are invertible; that is,
  \begin{enumerate}
    \item If $\nd^l_n \Sigma;\Gamma,\phi \andd \psi \seq \Delta$ then 
      $\nd^l_n\Sigma;\Gamma,\phi , \psi \seq \Delta$.
    \item If $\nd^l_n \Sigma;\Gamma,\phi \orr \psi \seq \Delta$ then
      $\nd^l_n\Sigma;\Gamma,\phi \seq \Delta$ and
      $\nd^l_n\Sigma;\Gamma,\psi \seq \Delta$.
    \item If $\nd^l_n \Sigma;\Gamma,\phi \impp \psi \seq \Delta$ then
      $\nd^l_n\Sigma;\Gamma,\psi \seq \Delta$.
    \item If $\nd^l_n \Sigma;\Gamma,\exists x.\phi\seq \Delta$ then
      $\nd^l_n\Sigma,y;\Gamma,\phi[y/x] \seq \Delta$.
    \item If $\nd^l_n \Sigma;\Gamma\seq \Delta,\forall x.\phi$ then
      $\nd^l_n\Sigma,y;\Gamma \seq \Delta,\phi[y/x]$.
    \item If $\nd^l_n\Sigma;\Gamma,\new \Aa{:}\nu .\phi \seq \Delta$ then
      $\nd^l_{n}\Sigma\#\Aa{:}\nu ;\Gamma,\phi \seq \Delta$ for fresh $\Aa\notin\Sigma$.
    \item If $\nd^l_n\Sigma;\Gamma \seq \Delta,\new \Aa{:}\nu .\phi$ 
      then  $\nd^l_n\Sigma\#\Aa{:}\nu ;\Gamma \seq \Delta,\phi$ for
      fresh $\Aa \notin\Sigma$.
  \end{enumerate}
\end{lemma}
\begin{proof}
  The proofs for the rules $\andL$, $\orL$, $\impL$, $\exL$, $\allR$
  are similar to those for the systems $\mathbf{G3c}$ and
  $\mathbf{G3im}$~\cite{negri01structural}.

  For $\newL$, the proof is by induction on the height of the derivation.  Most
  cases are straightforward.  Only cases such as $\allR,\exL,A_3,F$
  that introduce variables or name-symbols into $\Sigma$ are
  exceptions.  We show the reasoning for $\allR$.
  
  If the derivation is of the form
  \[\infer{\Sigma;\Gamma,\new \Aa{:}\nu.\phi \seq \forall
    x{:}\tau.\psi}{\Sigma,x{:}\tau;\Gamma,\new \Aa{:}\nu.\phi\seq \psi}\]
  then using the induction hypothesis, we have
  $\Sigma,x{:}\tau\#\Ab{:}\nu ;\Gamma,\phi(\Ab) \seq \psi$.  Using
  structural weakening we have $\Sigma\#\Aa{:}\nu,x{:}\tau\#\Ab{:}\nu
  ; \Gamma,\phi(\Ab) \seq \psi$.  Since $\Aa$ and $\Ab$ are fresh with
  respect to all terms in $Tm_\Sigma$, it is straightforward to show
  that $\Sigma\#\Aa{:}\nu,x{:}\tau\#\Ab{:}\nu :
  \Gamma,\swap{\Aa}{\Ab}{\phi(\Aa)} \seq \psi$.  Thus, by
  equivariance, we can derive
  $\Sigma\#\Aa{:}\nu,x{:}\tau\#\Ab{:}\nu;\Gamma,\phi(\Aa) \seq \psi$.
  Now $\Ab$ is not mentioned in the sequent so using $F$  we can
  derive $\Sigma\#\Aa{:}\nu,x{:}\tau ; \Gamma,\phi(\Aa) \seq \psi$,
  and using  $\allR$
  we can derive $\Sigma\#\Aa{:}\nu ; \Gamma,\phi(\Aa) \seq \forall
  x{:}\tau.\psi$, as desired.

  The proof for the invertibility of $\newR$ is symmetric.
\end{proof}
\begin{lemma}[Contraction]\labelLem{contraction}
  If $\nd^l_n\Sigma;\Gamma,\phi,\phi \seq \Delta$ is derivable then so
  is $\nd^l_n\Sigma;\Gamma,\phi \seq \Delta$.  Similarly, if
  $\nd^l_n\Sigma;\Gamma \seq \Delta,\phi,\phi$ is derivable then
  $\nd^l_n\Sigma;\Gamma \seq \Delta,\phi$.
\end{lemma}
\begin{proof}
  The proof is by induction on the lexicographic product of logical
  height and total height.  That is, the induction hypothesis
  applies to all derivations of smaller logical height and to all
  derivations of equal logical height but smaller total height.  Most
  cases are similar to any standard proof. The only new cases involve
  nonlogical rules and $\new \Aa{:}\nu.\phi$.  For the nonlogical rules it
  suffices to show that for each nonlogical rule that has a
  contractable instance, there is a nonlogical rule corresponding to
  the contraction.  The only such rule is $F_1$.  If the derivation is
  of the form
  \[\infer[F_1]{\Sigma;\Gamma,a\fresh x,a \fresh x \seq
    \Delta}{\Sigma;\Gamma,a\fresh x,a \fresh x,\swap{a}{a}{x} \eq x \seq \Delta}\]
  then we can transform the derivation to
  \[\infer[S_1]{\Sigma;\Gamma,a \fresh x \seq \Delta}{\Sigma;\Gamma,a
    \fresh x, \swap{a}{a}{x} \eq x \seq \Delta}\]

  Most of the remaining cases are standard.  The only interesting new
  case is when the contracted formula is derived using $\newL$:
  \[\infer[\newL]{\Sigma;\Gamma,\new \Aa{:}\nu.\phi(\Aa), \new \Ab{:}\nu.\phi(\Ab) \seq \Delta}{\Sigma\#\Aa{:}\nu;\Gamma,\phi(\Aa),\new \Ab{:}\nu.\phi(\Ab) \seq \Delta}\]
  Then using inversion we have $\nd^l_{n-1}\Sigma\#\Aa{:}\nu\#\Ab{:}\nu:
  \Gamma,\phi(\Aa),\phi(\Ab)\seq \Delta$.  Now using nonlogical rules we
  can derive $\nd^l_{n-1}\Sigma\#\Aa{:}\nu\#\Ab{:}\nu
  ;\Gamma,\phi(\Aa),\phi(\Aa)\seq \Delta$.  Then using the induction
  hypothesis we have $\nd^l_{n-1}\Sigma\#\Aa{:}\nu\#\Ab{:}\nu ;\Gamma,\phi(\Aa)\seq
  \Delta$.  Finally we can derive
\[\infer[\newL]{\Sigma;\Gamma,\new \Aa{:}\nu.\phi(\Aa) \seq \Delta}
{\infer[F]{\Sigma\#\Aa{:}\nu;\Gamma,\phi(\Aa) \seq \Delta}
{\Sigma\#\Aa{:}\nu\#\Ab{:}\nu;\Gamma,\phi(\Aa) \seq \Delta}}\]

The proof for right-contraction is symmetric, using the invertibility
of $\newR$.
\end{proof}

\subsection{Cut-Elimination}

As usual for sequent systems, the most important property to check to
verify that the system is sensible is cut-elimination.  

\begin{lemma}[Admissibility of Cut]
  If $\nd\Sigma;\Gamma \seq \Delta,\phi$ and
  $\nd\Sigma;\Gamma',\phi \seq \Delta$ then
  $\nd\Sigma;\Gamma,\Gamma' \seq \Delta,\Delta'$.  
\end{lemma}

\begin{proof}
  Following the proof of cut-elimination for similar systems such as
  $\mathbf{G3c}$ or $\mathbf{G3im}$ of~\cite{negri01structural}, we
  prove the lemma by induction on the structure of the cut-formula
  $\phi$ and then by a sub-induction on the sizes of the
  subderivations $\Pi$ of $\Sigma;\Gamma \seq \Delta,\phi$ and $\Pi'$
  of $\Sigma;\Gamma',\phi \seq \Delta$.  Thus, for the induction
  hypothesis, we may assume that the lemma holds for any instances
  with a less complex cut-formula or for all instances with the same
  cut-formula but with a smaller derivation of one or the other of
  $\Pi,\Pi'$.

  As in other proofs of cut-elimination for similar systems, there are
  four categories of cases:

  \begin{itemize}
  \item Base cases in which $\Pi$ or $\Pi'$ is an axiom or initial sequent.
    
  \item Left-commuting cases in which $\Pi$ starts with a rule in
    which $\phi$ is not principal.
  \item Right-commuting cases in which $\Pi'$ starts with a rule in
    which $\phi$ is not principal.
  \item Principal cases in which $\Pi$ and $\Pi'$ both start with a
    rule in which $\phi$ is principal.
  \end{itemize}

  All cases involving first-order rules exclusively are standard, and
  are shown in any standard proof of cut-elimination
  (e.g.~\cite{negri01structural} or~\cite{troelstra00basic}); their
  proofs rely upon the properties established in the previous section,
  including weakening, admissibility of $hyp^*$, contraction, and
  inversion.  In addition, Negri and von
  Plato~\cite{negri01structural} showed that nonlogical rules of the
  form we consider can be added to sequent systems like $\mathbf{G3c}$
  or $\mathbf{G3im}$ without damaging cut-elimination.  Hence, it will
  suffice to consider only the new cases involving the
  $\new$-quantifier rules.

  \begin{itemize}
  \item Base cases: There are no new base cases.
  \item Left-commuting cases: There are two new cases in which $\Pi$
    begins with $\newR$ or $\newL$.
    
    In the first case, we have
    \[\infer[\newL]{\Sigma;\Gamma,\new
      \Aa{:}\nu.\psi\seq\Delta,\phi}{\deduce{\Sigma\#\Aa{:}\nu;\Gamma,\psi \seq
        \Delta,\phi}{\Pi}}\]
    where $\Aa \not\in \Sigma$.  We can weaken $\Pi'$ to obtain
    a derivation $W(\Pi')$ of $\Sigma\#\Aa{:}\nu;\Gamma',\phi \seq \Delta'$, and by
    induction, we have $\Sigma\#\Aa{:}\nu;\Gamma,\psi,\Gamma' \seq
    \Delta,\Delta'$.  Then we may derive $\Sigma;\Gamma,\new
    \Aa{:}\nu.\psi,\Gamma' \seq \Delta,\Delta'$ using $\newL$.
      
    In the second case, we have
    \[\infer[\newR]{\Sigma;\Gamma\seq \Delta,\new \Aa{:}\nu.\psi
      ,\phi}{\deduce{\Sigma\#\Aa{:}\nu;\Gamma \seq \Delta,\psi,\phi}{\Pi}}\]
    where $\Aa \not\in \Sigma$.  We can weaken $\Pi'$ to
    get $W(\Pi')$ deriving $\Sigma\#\Aa{:}\nu;\Gamma',\phi\seq \Delta'$ and then by
    induction obtain $\Sigma\#\Aa{:}\nu;\Gamma',\Gamma \seq
    \Delta,\Delta',\psi$.  Using $\newR$ we can derive $\Sigma :
    \Gamma',\Gamma \seq \Delta,\Delta',\new \Aa{:}\nu.\psi$.

  \item Right-commuting cases.  These cases are exactly symmetric to
    the left-commuting cases.

    In the first case, we have
    \[\infer[\newL]{\Sigma;\Gamma',\phi,\new \Aa{:}\nu.\psi\seq\Delta'}{\deduce{\Sigma\#\Aa{:}\nu;\Gamma',\phi,\psi \seq
        \Delta'}{\Pi'}}\]
    where $\Aa \not\in \Sigma$.  We can weaken $\Pi$ to obtain a
    derivation $W(\Pi)$ of $\Sigma\#\Aa{:}\nu;\Gamma \seq \Delta,\phi$, and
    by induction, we have $\Sigma\#\Aa{:}\nu;\Gamma,\psi,\Gamma' \seq
    \Delta,\Delta'$.  Then we may derive $\Sigma;\Gamma,\new
    \Aa{:}\nu.\psi,\Gamma' \seq \Delta,\Delta'$ using $\newL$.
      
    In the second case, we have
    \[\infer[\newR]{\Sigma;\Gamma',\phi\seq \Delta',\new \Aa{:}\nu.\psi
      }{\deduce{\Sigma\#\Aa{:}\nu;\Gamma',\phi \seq \Delta',\psi}{\Pi'}}\]
    where $\Aa \not\in \Sigma$.  We can weaken $\Pi$ to obtain a
    derivation 
    $W(\Pi)$ of $\Sigma\#\Aa{:}\nu;\Gamma\seq \Delta,\phi$ and then by induction
    obtain $\Sigma\#\Aa{:}\nu;\Gamma',\Gamma \seq \Delta,\Delta',\psi$.
    Using $\newR$ we can derive $\Sigma ; \Gamma',\Gamma \seq
    \Delta,\Delta',\new \Aa{:}\nu.\psi$.
  \item Principal cases.  In this case, both $\Pi$ and $\Pi'$ decompose
    the cut formula.  The only new rule for decomposing formulas on
    the right is $\newR$, so the only new principal cut case is when
    we have
    \[\infer[\newR]{\Sigma;\Gamma\seq \Delta,\new \Aa{:}\nu.\phi}{\deduce{\Sigma\#\Aa{:}\nu;\Gamma\seq \Delta,\phi}{\Pi}} \quad \infer[\newL]{\Sigma;\Gamma',\new\Aa{:}\nu.\phi \seq \Delta'}{\deduce{\Sigma\#\Aa{:}\nu;\Gamma',\phi \seq \Delta'}{\Pi'}}\]
    for some $\Aa \not\in \Sigma$.  By induction we have
    $\Sigma\#\Aa{:}\nu;\Gamma,\Gamma' \seq \Delta,\Delta'$, and we may
    conclude $\Sigma;\Gamma,\Gamma' \seq \Delta,\Delta'$ by an
    application of the freshness rule.
  \end{itemize}
This completes the proof.
\end{proof}

\begin{theorem}
  Any derivable \NLseq sequent has a cut-free derivation; there is an
  algorithm for producing such derivations.
\end{theorem}
\begin{proof}
  Proof by induction on the number of cuts.  Given a derivation using
  cut, we can always find an uppermost use of cut in the derivation
  tree and remove it. This reduces the number of cuts by one.
\end{proof}

\subsection{Intuitionistic calculus}

\emph{Intuitionistic \NLseq} (\INLseq) is based on the
multiple-conclusion intuitionistic calculus
$\mathbf{G3im}$~\cite{negri01structural}, in which certain rules are
restricted to discard alternative conclusions (see \refFig{G3im}).  It
is straightforward to show that all of the structural
properties including cut-elimination hold for \INLseq; the same
arguments as given above in the classical case apply.  We will show
in \refSec{intuitionistic-conservativity} that \INLseq corresponds to
a theory of first-order intuitionistic logic that is equivalent to
Pitts' axiomatization in classical \NL. 
\begin{theorem}
  In \INLseq, if $\Sigma;\Gamma \seq \Delta$ holds then there is a
  cut-free derivation of $\Sigma;\Gamma \seq \Delta$.
\end{theorem}
It is also straightforward to
show that \INLseq is equivalent to a single-conclusion intuitionistic
calculus, since the nonlogical and $\new$-quantifier rules preserve
the single-conclusion property.  

\begin{theorem}
If $\Sigma;\Gamma \seq \Delta$ holds in
\INLseq  then $\Sigma;\Gamma \seq \Orr \Delta$ holds in the
single-conclusion variant of \INLseq.
\end{theorem}
\begin{proof}
Most cases of the proof are analogous to the usual proof relating $\mathbf{G3i}$
and $\mathbf{G3im}$~\cite{negri01structural}.  The additional cases involve the nonlogical and
$\new$-quantifier rules.  Of these, the nonlogical rules are
straightforward because nothing changes on the right-hand side of the
sequent in these rules.  The case for $\newL$ is also straightforward
for the same reason.  

We show the case for $\newR$.  Suppose the derivation is of the form:
\[\infer[\newR]{\Sigma;\Gamma \seq \new
  \Aa{:}\nu.\phi,\Delta}{\Sigma\#\Aa{:}\nu;\Gamma \seq \phi,\Delta}
\]
By induction on the subderivation we know that
$\Sigma\#\Aa{:}\nu;\Gamma \seq \phi \orr \Orr\Delta$.  
We reason as follows:
\[\infer[F]{\Sigma;\Gamma \seq \new \Aa{:}\nu.\phi \orr \Orr \Delta}{
\infer[cut]{\Sigma\#\Aa{:}\nu; \Gamma \seq \new \Aa{:}\nu.\phi \orr \Orr \Delta}{
\Sigma\#\Aa{:}\nu;\Gamma \seq \phi \orr \Orr\Delta
& 
\infer[\orL]{\Sigma\#\Aa{:}\nu; \Gamma, \phi \orr \Orr\Delta \seq \new
  \Aa{:}\nu.\phi \orr \Orr \Delta}{
\infer[\orR_1]{\Sigma\#\Aa{:}\nu; \Gamma, \phi \seq \new
  \Aa{:}\nu.\phi \orr \Orr \Delta}
{
\infer[\newR]{\Sigma\#\Aa{:}\nu; \Gamma, \phi \seq \new
  \Aa{:}\nu.\phi}{
\Sigma\#\Aa{:}\nu\#\Ab{:}\nu; \Gamma, \phi \seq \phi[\Ab/\Aa]
}
}
& 
\infer[\orR_2]{\Sigma\#\Aa{:}\nu; \Gamma, \Orr \Delta \seq \new
  \Aa{:}\nu.\phi \orr \Orr \Delta
}{\infer[hyp^*]{\Sigma\#\Aa{:}\nu; \Gamma, \Orr \Delta \seq \Orr \Delta}{}}
}
}
}
\]
We can use the intuitionistic (single-conclusion) variant of
\refLem{swapping-fresh-names} to conclude
$\Sigma\#\Aa{:}\nu\#\Ab{:}\nu; \Gamma, \phi \seq \phi[\Ab/\Aa]$.
\end{proof}
\begin{figure*}[tb]
\[\begin{array}{cc}
\infer[\impR]{\Sigma;\Gamma \seq \phi \impp \psi,\Delta}{\Sigma;\Gamma,\phi \seq\psi}
& 
\infer[\impL]{\Sigma;\Gamma,\phi \impp \psi \seq \Delta}
      {\Sigma;\Gamma, \phi \impp \psi \seq \phi& 
        \Sigma;\Gamma,\psi \seq \Delta}\smallskip\\
\infer[\allR]{\Sigma;\Gamma \seq \forall x{:}\sigma.\phi,\Delta}
        {\Sigma,x{:}\sigma;\Gamma \seq \phi & 
        (x \not\in \Sigma)}
&
\infer[\allL]{\Sigma;\Gamma, \forall x{:}\sigma.\phi\seq \Delta}
        {\Sigma \vdash t : \sigma & 
        \Sigma;\Gamma, \forall x{:}\sigma.\phi,\phi\subs{t}{x} \seq \Delta}
\smallskip\\
\infer[\exR]{\Sigma;\Gamma \seq \exists x{:}\sigma.\phi,\Delta}
        {\Sigma \vdash t : \sigma
&       \Sigma;\Gamma \seq \exists x{:}\sigma.\phi,\phi\subs{t}{x},\Delta}
&
\infer[\exL]{\Sigma;\Gamma,\exists x{:}\sigma.\phi \seq \Delta}
        {\Sigma,x{:}\sigma;\Gamma, \phi\seq \Delta & 
          (x \not\in \Sigma)}\smallskip\\
\end{array}\]
\hrule
\caption{Variant rules for the intuitionistic multiple-conclusion
  calculus ($\mathbf{G3im}$)}\labelFig{G3im}
\end{figure*}

\section{Applications}\labelSec{applications}

\subsection{Syntactic Consistency}

For pure first-order logic, cut-elimination immediately implies
consistency, since by inspection of the rules there can be no shortest
proof of $\cdot;\cdot\seq \bot$.  However, in the presence of general nonlogical
rules, only a weaker result holds.  We say that an atomic formula is a
\emph{constraint} if it is an equality or freshness formula, and
$\Gamma$ is a constraint set of it contains only constraints.
\begin{proposition}\labelProp{partial-consistency}
  If $\cdot;\cdot\seq\bot$ has a cut-free derivation, then it has one using only
  nonlogical rules, in which each sequent is of the form $\cdot; \Gamma \seq
  \bot$, where $\Gamma$ is a constraint set.
\end{proposition}
The proof is immediate by observing that only nonlogical rules are
applicable to a derivation of $\cdot;\Gamma \seq \bot$ where $\Gamma$ is a
constraint set.  In particular, note that the instance of the $Ax$
rule scheme for $a \fresh a \impp \bot$ (axiom $F_3$) has no hypotheses:
\[
\infer[F_3]{\Sigma;\Gamma ,a\fresh a\seq \Delta}{}\]
so it is not
necessary to allow $\bot$ as a constraint (though this would not do
any harm either).

This means that nominal logic is consistent if and only if the
nonlogical rules are consistent.  We know that classical nominal logic
is consistent with respect to the semantics given by Pitts using
nominal sets~\cite{pitts03ic}, and we will show in the next section
that the two systems are equivalent, however, here we would like to
give a direct syntactic proof that applies to both classical and
intuitionistic variants of \NLseq.  To prove the consistency of the
nonlogical rules, it is necessary to exhibit a model.  We review how
to define a Herbrand-style semantics in terms of the syntax of nominal
terms (see e.g. Cheney~\cite{cheney06jsl} for more details).

\begin{definition}[Syntactic Swapping, Equality and Freshness]
Let $Tm$ be the set of swapping-free nominal terms generated by the grammar
\[t ::= \Aa \mid c \mid f(\vec{t}) \mid \abs{\Aa}{t}\]
We define the \emph{swapping function} on such terms as follows:
\begin{eqnarray*}
  \swap{\Aa}{\Ab}{\Aa} &=& \Ab\\
  \swap{\Aa}{\Ab}{\Ab} &=& \Aa\\
  \swap{\Aa}{\Ab}{\Ac} &=& \Ac \quad (\Aa, \Ab \neq \Ac)\\
  \swap{\Aa}{\Ab}{c} &=& c
\\
  \swap{\Aa}{\Ab}{f(\vec{t})} &=& f(\swap{\Aa}{\Ab}{\vec{t}})\\
  \swap{\Aa}{\Ab}{\abs{\Ac}{t}} &=& \ab{\swap{\Aa}{\Ab}{\Ac},\swap{\Aa}{\Ab}{t}}\end{eqnarray*}

We define the \emph{freshness} relation on ground terms
using the rules:
\[\begin{array}{c}
\infer{\Aa \fresh \Ab}{(\Aa \neq \Ab)}\quad
\infer{\Aa \fresh c}{}\quad
\infer{\Aa \fresh f(\vec{t})}{\Aa \fresh t_1 & \ldots & \Aa\fresh t_n}\quad
\infer{\Aa \fresh \abs{\Aa}{t}}{}\quad
\infer{\Aa \fresh \abs{\Ab}{t}}{\Aa \fresh t & (\Aa \neq \Ab)}\\
\end{array}\]

The \emph{nominal equality} relation is defined as follows:
\[\begin{array}{c}
\hyp{\Aa \eq \Aa}\quad
\hyp{c \eq c}\quad
\infer{f(\vec{t}) \eq f(\vec{u}) }{t_1 \eq u_1 &\ldots& t_n \eq u_n}\quad
\infer{\abs{\Aa}{t} \eq \abs{\Aa}{u}}{t \eq u}\quad
\infer{\abs{\Aa}{t} \eq \abs{\Ab}{u}}{t \eq \swap{\Aa}{\Ab}{u} & \Aa \fresh u & (\Aa \neq \Ab)}
\end{array}\]
\end{definition}

The following properties of syntactic freshness and equality are a
special case of more general properties established elsewhere, e.g. by
Urban et al.~\cite{urban04tcs}:
\begin{proposition}
  The nominal equality relation $\eq$ is an equivalence relation.
  Hence, $NTm = Tm/_{\eq}$ is well-defined.  Moreover, both $\eq$ and
  $\fresh$ are equivariant relations on $Tm$.
\end{proposition}

We now show how to interpret arbitrary nominal terms in $NTm$. 
\begin{definition}
  Let $\theta : V \to NTm$ be a substitution of ground nominal terms
  for variables, called an \emph{interpretation}.  We lift $\theta$ to
  a function from arbitrary terms to $NTm$ as follows:
  \begin{eqnarray*}
    \theta(\Aa) &=& \Aa\\
    \theta(c) &=&c\\
    \theta(f(\vec{t})) &=& f(\theta(t_1),\ldots,\theta(t_n))\\
    \theta(\swap{a}{b}{t}) &=& \swap{\theta(a)}{\theta(b)}{\theta(t)}\\
    \theta(\abs{a}{t}) &=& \abs{\theta(a)}{\theta(t)}
  \end{eqnarray*}

  We say that $\theta : FV(\Sigma) \to NTm$ \emph{satisfies} $\Sigma$
  (written $\theta : \Sigma$) if $\theta(x) : \Sigma(x)$ for each $x$
  and $\Aa \fresh \theta(x)$ for each constraint $\Aa \fresh x \in
  |\Sigma|$.  

  We write $\theta \models t \eq u$ or $\theta \models a \fresh t$ to
  indicate that $\theta(t) \eq \theta(u)$ or $\theta(a) \fresh
  \theta(t)$ respectively.  Similarly, $\theta \models \Gamma$
  indicates that $\theta \models A$ for each constraint $A$ in
  constraint set $\Gamma$.  We say that a constraint $A$ (or
  constraint set $\Gamma$) is \emph{satisfiable} if there is an
  interpretation $\theta : \Sigma$ such that $\theta\models A$
  (respectively, $\theta\models \Gamma)$ holds in $NTm$.
\end{definition}

\begin{proposition}\label{prop:nonlogical-valid}
  The axioms listed in \refFig{nonlogical-axioms} are valid for $NTm$, in the
  sense that for each axiom $\Andd P \impp \Orr Q$, if $\theta \models
  \Andd P$ then $\theta \models Q_i$ for some $Q_i \in \Orr Q$.
\end{proposition}
\begin{proof}
  For $S_1$ and $S_2$, the proof is by induction on the definition of
  swapping for ground terms.  The validity of $S_3$ is immediate.

  For the equivariance axioms, the definition of swapping makes plain
  that abstraction and other function symbols besides
  swapping itself are equivariant.  In addition, it is not difficult
  to show that
  \[\swap{\Aa}{\Aa'}{\swap{\Ab}{\Ab'}{x}} =
  \swap{\swap{\Aa}{\Aa'}{\Ab}}{\swap{\Aa}{\Aa'}{\Ab'}}\swap{\Aa}{\Aa'}{x}\] 
  that is, that the syntactic swapping function is equivariant.  For
  the equivariance axioms for formulas, we only need to consider
  $E_{\eq}$ and $E_{\fresh}$.  But clearly equality is equivariant
  since
  \[x \eq y \impp \swap{\Aa}{\Ab}{x} \eq \swap{\Aa}{\Ab}{y}\] can be shown by
  induction on the derivation of $x \eq y$; similarly,
  \[\Aa \fresh x \impp \swap{\Ab}{\Ab'}{\Aa} \fresh \swap{\Ab}{\Ab'}{x}\]
  can be shown valid by induction on the derivation of $\Aa \fresh x$.

  For the axiom $F_1$, we must show that if $\Aa \fresh x$ and
  $\Ab\fresh x$ then $\swap{\Aa}{\Ab}{x} \eq x$.  The proof is by
  induction on the structure of $x$.  For $x = c$ the result is
  immediate; similarly, for $x = f(y\vec{t}$  the
  induction step is straightforward.  For $x = \Ac$, we have $\Aa,
  \Ab \neq \Ac$ so $\swap{\Aa}{\Ab}{\Ac} = \Ac \eq \Ac$.  For $x =
  \abs{\Ac}{y}$, there are two cases.  If $\Aa,\Ab \neq \Ac$ then we
  have $\Aa,\Ab \fresh y$ and
  \[\swap{\Aa}{\Ab}{\abs{\Ac}{y}} = \abs{\swap{\Aa}{\Ab}\Ac}{\swap{\Aa}{\Ab}{y}} \eq \abs{\Ac}{y}\]
  since by induction $\swap{\Aa}{\Ab}{y} \eq y$.  Otherwise, without
  loss of generality suppose $\Ab = \Ac$ (the case where $\Aa = \Ac$
  is symmetric).  We need to show that $\swap{\Aa}{\Ab}{\abs{\Ab}{y}} \eq
  \abs{\Ab}{y}$, or equivalently that $\abs{\Aa}{\swap{\Aa}{\Ab}{y}} \eq
  \abs{\Ab}{y}$.  If $\Aa = \Ab$, this is trivial.  Otherwise, it is
  sufficient to show that $\swap{\Aa}{\Ab}{y} \eq \swap{\Aa}{\Ab}{y}$
  (which is immediate) and $\Aa \fresh y$.  But since $\Aa \fresh
  \abs{\Ab}{y}$ and $\Aa \neq \Ab$, we know that $\Aa \fresh y$ holds.

  For $F_2$, clearly any two name symbols $\Aa{:}\nu$ and $\Ab{:}\nu'$
  of different sorts are distinct, so $\Aa \fresh \Ab$.

  For $F_3$, we need to show that $\Aa \fresh \Aa$ is not derivable.
  This is immediate from the definition of the freshness relation.

  For $F_4$, we need to show that either $\Aa \fresh \Ab$ or $\Aa \eq
  \Ab$ is derivable.  If $\Aa = \Ab$ then $\Aa \eq \Ab$ is derivable;
  otherwise $\Aa \neq \Ab$ so $\Aa \fresh \Ab$ is derivable.

  Finally, for $A_1$ we need to show that if $\Aa \fresh y$ and $x \eq
  \swap{\Aa}{\Ab}{y}$ then $\abs{\Aa}{x} \eq \abs{\Ab}{y}$.  There are
  two cases.  If $\Aa \neq \Ab$ then the last rule in the definition
  of nominal equality applies to show $\abs{\Aa}{x} \eq \abs{\Ab}{y}$.
  Otherwise, $\Aa = \Ab$ so $x \eq \swap{\Aa}{\Ab}{y} = y$ and so 
  $\abs{\Aa}{x} \eq \abs{\Ab}{y}$.
\end{proof}

\begin{proposition}\labelProp{cases-abstractions-equal}
  If $\theta \models \abs{a}{x} \eq \abs{b}{y}$ then either $\theta
  \models a \eq b ,x \eq y$ or $\theta \models a \fresh y, x \eq
  \swap{a}{b}{y}$.
\end{proposition}
\begin{proof}
  The proof is by case analysis of the possible derivations of $\theta
  (\abs{a}{x}) \eq \theta(\abs{b}{y})$.  There are only two cases,
  corresponding to the last two rules in the definition of structural
  equality.  The result is immediate.
\end{proof}

\begin{proposition}\labelProp{substitution-freshness}
  If $\theta :\Sigma$ then $\theta \models \Aa \fresh t$ for each $\Aa
  \fresh t \in |\Sigma|$.
\end{proposition}
\begin{proof}
  The proof is by induction on the structure of $t$.
  The critical case is for $t$ a variable; in this case, we need to
  use the fact that $\theta:\Sigma$ only if $\Aa \fresh \theta(x)$ for
  each $\Aa \fresh x \in |\Sigma|$.
\end{proof}

\begin{theorem}\labelThm{unsatisfiable-derives-false}
  Let $\Gamma$ be a set of freshness and equality formulas.  If
  $\Sigma;\Gamma \seq \bot$ is derivable then $\Gamma$ is
  unsatisfiable.
\end{theorem}
\begin{proof}
  Proof is by induction on the structure of the derivation.  Note that
  the only applicable rules are nonlogical rules.  There is one case
  for each nonlogical rule.  Most cases are straightforward.  We
  present some interesting cases.
  
  All of the axioms in \refFig{nonlogical-axioms} hold in $NTm$, by
  Proposition~\ref{prop:nonlogical-valid}, so the cases in which these
  axioms are used are straightforward.  For example, for a derivation
  of the form
  \[\infer[F_3]{\Sigma;\Gamma,a\fresh a \seq \bot}{}\]
  clearly $\Gamma,a \fresh a$ is unsatisfiable.

  For a derivation of the form
  \[\infer[F_4]{\Sigma;\Gamma \seq \bot}{\Sigma;\Gamma, a \fresh b\seq
    \bot&\Sigma;\Gamma,a \eq b \seq \bot}\]
  we have $\Gamma,a \eq b$ and $\Gamma,a \fresh b$ unsatisfiable.  If
  $\theta: \Sigma$ then either $\theta(a) \eq \theta(b)$ or $\theta(a)
  \neq \theta(b)$, in which case $\theta(a) \fresh \theta(b)$.  In
  either case, $\theta$ cannot satisfy $\Gamma$.
  
  For a derivation ending with $F$, 
  \[\infer[F]{\Sigma;\Gamma\seq \bot}{\Sigma\#\Aa{:}\nu;\Gamma \seq \bot}\]
  if $\theta :\Sigma$, then without loss of generality we can assume
  $\Aa \fresh \theta$ so that $\theta:\Sigma\#\Aa{:}\nu$ and so $\theta
  \not\models\Gamma$ by induction.

  For $\Sigma\#$:
  \[\infer[\Sigma\#]{\Sigma;\Gamma\seq \bot}{\Sigma;\Gamma,\Aa\fresh
    t\seq \bot &  (\Aa \fresh t \in |\Sigma|)
    }\]
  if $\theta:\Sigma$ then $\theta \models
  \Aa \fresh t$ for any $\Aa \fresh t \in |\Sigma|$, by \refProp{substitution-freshness}.
  Consequently $\theta \not\models \Gamma$.

  For $A_2$, 
  \[\infer[A_2]{\Sigma;\Gamma,\abs{a}{x} \eq \abs{b}{y} \seq
    \bot}{\Sigma;\Gamma,a \eq b ,x \eq y \seq \bot&\Sigma;\Gamma,a
    \fresh y, x \eq \swap{a}{b}y \seq \bot}\]
  suppose $\theta:\Sigma$.  By induction $\theta \not\models \Gamma,a
  \eq b ,x \eq y$ and $\theta \not\models\Gamma ,a \fresh y, x \eq
  \swap{a}{b}y$.  There are three cases.  If $\theta(a) \eq \theta(b)$
  and $\theta(x) \eq \theta(y)$, then $\theta \not\models \Gamma$.
  Similarly, if $\theta(a) \fresh \theta(y)$ and $\theta(x) \eq
  \swap{\theta(a)}{\theta(b)}\theta(y)$ then $\theta\not\models
  \Gamma$.  Otherwise, by the contrapositive of
  \refProp{cases-abstractions-equal}, $\theta \not\models \abs{a}{x}
  \eq \abs{b}{y}$.  In any case, $\theta \not\models \Gamma,\abs{a}{x}
  \eq \abs{b}{y}$.

For $A_3$,
\[\infer[A_3]{\Sigma;\Gamma\seq \bot}{\Sigma\nd t:\abs{\nu}{\tau} & \Sigma,a{:}\nu,x{:}\tau;\Gamma,t\eq \abs{a}{x} \seq \bot}\]
if $\theta :\Sigma$ then $\theta(t) = \abs{\Aa}{v}$ for some $\Aa:\nu$
and $t:\tau$, so let $\theta' = \theta[a \mapsto \Aa,x \mapsto t]$.
Clearly $\theta' : \Sigma,a{:}\nu,x{:}\tau$ and $\theta' \models t \eq \abs{a}{x}$
so by induction $\theta' \not\models \Gamma$.  Since $\Gamma$ does not
mention $a$ or $x$, we can conclude $\theta \not\models \Gamma$.
\end{proof}

\begin{corollary}[Syntactic consistency]
  There is no derivation of $\cdot;\cdot\seq \bot$.
\end{corollary}
\begin{proof}
  This follows from \refProp{partial-consistency} and
  \refThm{unsatisfiable-derives-false}, since $\emptyset$ is a
  satisfiable constraint set.
\end{proof}

\subsection{Orthogonality of abstraction}

Using cut-elimination, we can also show that some parts of the
equational theory are ``orthogonal extensions'', that is, derivable
sequents not mentioning abstraction 
can be
derived without using the special properties of these symbols.

\begin{theorem}[Conservativity]\labelThm{conservativity-abs}
  Suppose $\Sigma$ has no variables mentioning abstraction
 and $\Sigma;\Gamma \seq \Delta$ and $\Gamma, \Delta$ have no
  subterms of the form $\abs{a}{t}$.
  Then there is a derivation of $\Sigma;\Gamma \seq \Delta$ that does
  not use any nonlogical rules involving abstraction.
\end{theorem}
\begin{proof}
  We say that a context, formula, formula multiset, or sequent is
  abstraction-free if the abstraction function symbol and type
  constructor do not appear in it.  A derivation is abstraction-free
  if the rules $A_1,A_2,A_3$ do not appear in it.  We write
  $\nd^{-A}$ for abstraction-free derivability.

  The proof is by induction on the structure of cut-free derivations.
  We need a stronger induction hypothesis.  We say $\Gamma$ is
  \emph{good} if abstraction is only mentioned in equations and
  freshness formulas.  Note that if $\Sigma$ is abstraction-free and
  there are no constants whose types mention abstraction then the only
  well-formed closed terms of type $\abs{\nu}{\tau}$ are of the form
  $\abs{a}{t}$.  Hence, any equations among abstraction-typed terms
  are of the form $\abs{a}{t} \eq \abs{b}{u}$; we call such formulas
  abstraction equations.  Any context can be partitioned into
  $\Gamma,\Gamma'$ such that $\Gamma'$ contains all the abstraction
  equations.  We say that $\Gamma'$ is \emph{redundant} relative to $\Gamma$
  if whenever $\abs{a}{t} \eq \abs{b}{u} \in \Gamma'$, we have either
  $\nd^{-A}\Sigma;\Gamma \seq a \eq b$ and $t \eq u$ or
  $\nd^{-A}\Sigma;\Gamma \seq a \fresh u$ and $t \eq \swap{a}{b}{u}$.
  
  We will show that if $\Sigma,\Delta$ are abstraction-free and
  $\Gamma,\Gamma'$ is good and $\Gamma'$ is redundant relative to
  $\Gamma$, then if $\nd\Sigma;\Gamma,\Gamma' \seq \Delta$ then
  $\nd^{-A}\Sigma;\Gamma \seq \Delta$.  An abstraction-free $\Gamma$
  is obviously good and redundant relative to $\emptyset$, so the
  main theorem is a special case.

  The proof is by structural induction on the derivation.  The cases
  involving left or right rules are straightforward because such rules
  act only on $\Gamma$ and do not affect goodness and redundancy.  The
  case for $hyp$ is easy since the hypothesis cannot be in $\Gamma'$.

  For $A_1$, we have 
  \[\infer[A_1]{\Sigma;\Gamma,a\fresh x, x \eq \swap{a}{b}{y},\Gamma' \seq \Delta}{\Sigma;\Gamma,a\fresh x, x \eq \swap{a}{b}{y},\Gamma',\abs{a}{x} \eq \abs{b}{y}\seq \Delta}\]
  Clearly, $\Gamma',\abs{a}{x} \eq \abs{b}{y}$ is redundant relative
  to $\Gamma,a\fresh x, x \eq \swap{a}{b}{y}$.  Also,
  goodness is preserved.  So by induction we have
  $\Sigma;\Gamma,a\fresh x, x \eq \swap{a}{b}{y} \seq \Delta$, as
  desired.

  For $A_2$, we have
  \[\infer[A_2]{\Sigma;\Gamma,\Gamma',\abs{a}{x} \eq \abs{b}{y} \seq \Delta}
  {\Sigma;\Gamma,\Gamma',\abs{a}{x} \eq \abs{b}{y},a \eq b, x \eq y
    \seq \Delta & \Sigma;\Gamma,\Gamma',\abs{a}{x} \eq \abs{b}{y},a
    \fresh y,x \eq \swap{a}{b}{y} \seq \Delta}\] 
Since  $\Gamma',\abs{a}{x} \eq \abs{b}{y}$  is
  redundant relative to  $\Gamma$, there are
  two cases.  If $\Sigma;\Gamma \seq a \eq b$ and $ x \eq y$, then by
  induction we have a derivation of $\Sigma;\Gamma,a \eq b, x \eq y
  \seq \Delta$, and using cut we can derive $\Sigma;\Gamma \seq
  \Delta$ as desired.  Otherwise, if $\Sigma;\Gamma \seq a \fresh y$ and $ x
  \eq \swap{a}{b}y$, then by induction we have a derivation of
  $\Sigma;\Gamma, a \fresh y,x \eq \swap{a}{b}y\seq \Delta$, and using
  cut we can derive $\Sigma;\Gamma \seq \Delta$ as desired.
  Cut-elimination does not introduce uses of the abstraction rules, so
  the resulting derivations are abstraction-free.

  For $A_3$, we have
  \[\infer[A_3]{\Sigma;\Gamma,\Gamma' \seq \Delta}{\Sigma\nd t:\abs{\nu}{\tau} & \Sigma,a{:}\nu,x{:}\tau;\Gamma,t \eq \abs{a}{x},\Gamma' \seq \Delta}\]
  Since $\Sigma$ has no variables of abstraction type, we must have $t
  = \abs{u}{v}$ for some terms $\Sigma \nd u:\nu$ and $\Sigma \nd v:\tau$.  Therefore,
  we can substitute into the derivation $\Sigma,a{:}\nu,x{:}\tau;\Gamma,\Gamma',t
  \eq \abs{a}{x} \seq \Delta$ to get $\Sigma;\Gamma,\Gamma',\abs{u}{v}
  \eq \abs{u}{v}\seq \Delta$.  Clearly $\Sigma;\Gamma \seq u \eq u$
  and $v \eq v$, and $\Gamma',\abs{u}{v} \eq \abs{u}{v}$ is redundant
  relative to $\Gamma$, so by induction, we have a derivation of
  $\Sigma;\Gamma \seq \Delta$.

  For the reflexivity rule ${\eq}R$, we have
  \[\infer[{\eq}R]{\Sigma;\Gamma,\Gamma'\seq \Delta}{\Sigma;\Gamma,\Gamma',t\eq t\seq \Delta}\]
  If $t = \abs{a}{x}$, then clearly $\Gamma \seq a \eq a$ and $x \eq
  x$, so $\Gamma', \abs{a}{x} \eq \abs{a}{x}$ is redundant relative to
  $\Gamma$, and we have $\Sigma;\Gamma \seq \Delta$ by induction.
  Otherwise, $\Gamma,\Gamma',t\eq t$ is obviously still good and
  $\Gamma'$ redundant with respect to $\Gamma,t\eq t$, so we can again
  conclude $\Sigma;\Gamma \seq \Delta$ by induction.

  For ${\eq}S$-derivations, we have
  \[\infer[{\eq}S]{\Sigma;\Gamma,\Gamma',t\eq u, P(t) \seq \Delta}{\Sigma;\Gamma,t\eq u, P(t),P(u) \seq \Delta}\]
  If $P(u)$ is not an equation among abstraction-typed terms then the
  induction step is easy.  There are many cases depending on the
  structure of $P(x)$, but in each case we can show that $P(u)$ is
  also redundant relative to $\Gamma,t \eq u$ (if $t \eq u$ is not an
  abstraction equation) or $\Gamma$ (if $t \eq u$ is an abstraction
  equation).
  
  The remaining nonlogical rules do not involve formulas of the form
  $\abs{a}{x} \eq \abs{b}{y}$, so the induction step is immediate for
  these rules.

\end{proof}

\subsection{Equivalence to Nominal Logic}

In this section we discuss the relationship between the sequent
calculi \NLseq and \INLseq and classical and intuitionistic variants of
Nominal Logic respectively.  We aim to show that, modulo a
straightforward syntactic translation, formulas are provable in one
system if an only if they are provable in the other.  This in turn
suggests that they are equally expressive in a model-theoretic sense
(provided models for \NLseq are defined in an appropriate way for its
slightly different syntax, as done for example for
\FL~\cite{gabbay07jal}); however, in this article we will not pursue
the model theory of \NLseq.

\subsubsection{Classical Nominal Logic}
We first consider the classical case.  We write $NL$ for the set of
all axioms of Pitts' axiomatization of nominal logic, as reviewed in
\refSec{background-nl}.  For ordinary variable contexts $\Sigma$ and
$NL$-formula multisets $\Gamma,\Delta$, we write
$\nd_{NL}\Sigma;\Gamma \seq \Delta$ to indicate that
$\Sigma;\Gamma,\Gamma' \seq_{\mathbf{G3c}}  \Delta$ for some $\Gamma' \subseteq NL$.
Without loss of generality, a finite $\Gamma'$ can always be used.
We write $\nd_{\NLseq}$ for derivability in \NLseq.

There is one technical point to
address.  Our system contains explicit name-constants quantified by
$\new$ and appearing in typing contexts, whereas in Pitts' system
$\new$ quantifies ordinary variables.  To bridge this gap, we
translate \NL formulas to \NLseq formulas by replacing $\new$-bound
variables with fresh name-symbols.  For example, the \NL formula $\new
a{:}\nu.\new b{:}\nu'.  p(a,b)$ translates to the \NLseq formula $\new \Aa{:}\nu.\new
\Ab{:}\nu'.p(\Aa,\Ab)$.  We write $\phi^*$ for the translation of $\phi$, which is defined as follows:
\begin{eqnarray*}
A^* &= &A\\
\false^* &=& \false\\
(\phi \impp \psi)^*&=& \phi^* \impp \psi^*\\
(\forall x{:}\tau.\phi)^* &=& \forall x{:}\tau.\phi^*\\
(\new a{:}\nu.\phi)^* &=& \new \Aa{:}\nu. (\phi^*[\Aa/a]) \quad (\Aa = \iota(a))
\end{eqnarray*}
Technically, we translate sequents or derivation mentioning variables
in $\Var \cup \Name'$, to sequents or derivations mentioning variables
in $\Var \cup \Name'$, where $\Name'$ is an isomorphic copy of the set
of names $\Name$.  We assume that before translation, formulas are
renamed so that $\new$-bound variables are in $\Name'$, and we fix an
isomorphism $\iota : \Name' \to \Name$. In what follows, we will
sometimes leave $\iota$ implicit and assume that $\iota(a) = \Aa$
whenever we encounter a $\new$-quantifier or context of the form
$\Sigma\#\Aa{:}\nu$.

The
omitted cases for $\true,\andd,\orr,\exists$ are derivable via de
Morgan identities.  The translation of a judgment $\Sigma;\Gamma \seq
\Delta$ is $\Sigma;\Gamma^* \seq \Delta^*$, where $\Gamma^*, \Delta^*$
is the result of translating each element of $\Gamma, \Delta$
respectively.

We first show that every theorem of $NL$ translates to a theorem of
\NLseq.  
\begin{theorem}
  If $\nd_{NL}\Sigma;\Gamma \seq \Delta$ then $\nd_{\NLseq}
  \Sigma;\Gamma^*\seq \Delta^*$.
\end{theorem}
\begin{proof}
  We defined $\nd_{\NL} \Sigma;\Gamma \seq
  \Delta$ to mean $\nd_{\mathbf{G3c}}\Sigma;\Gamma,\Gamma' \seq
  \Delta$ for some finite subset $\Gamma' \subseteq \NL$.  Any
  $\mathbf{G3c}$ derivation is an $\NLseq$ derivation, so we just need
  to show that in $\NLseq$, all of the uses of \NL axioms are
  redundant.  We will show that each axiom $\phi \in \NL$ is derivable
  in \NLseq.  Thus, using $cut$ finitely many times, we can derive
  $\Sigma;\Gamma \seq \Delta$ in \NLseq.

  For most of the axioms, this is straightforward.  All of the axioms
  of the form $\forall \vec{x}.\Andd \vec{P} \impp \Orr \vec{Q}$ are
  clearly derivable from the corresponding nonlogical rules as
  follows:
  \[\infer[\allR]{\cdot;\cdot \seq \forall \vec{x}{:}\vec{\tau}.\Andd \vec{P} \impp \Orr
    \vec{Q}}{\infer[\impR,\andR]{\vec{x}{:}\vec{\tau};\cdot \seq \Andd \vec{P} \impp \Orr
      \vec{Q}}{\infer[Ax]{\vec{x}{:}\vec{\tau} ; \vec{P} \seq \Orr
        \vec{Q}}{\vec{x}{:}\vec{\tau};\vec{P},Q_1 \seq \Orr \vec{Q} & \cdots
        &\vec{x}{:}\vec{\tau};\vec{P},Q_n \seq \Orr \vec{Q} }}}\]
  with the topsequents all derivable using $\orR$ and $hyp$.

  This leaves axioms not fitting this pattern, including $(CF_2)$, $(CF_4)$,
  $(CA_1)$, $(CA_2)$, and $(CQ)$.  $(CA_1)$ and $(CA_2)$ can be derived using the
  nonlogical rules $A_1,A_2,A_3,{\eq}S$ of \NLseq, and $(CF_2)$ using
  $F_3$ and $F_4$ of \NLseq.
  We will show the cases for $(CF_4)$ and both directions of $(CQ)$ in
  detail.

  For an instance $\forall \vec{x}.\exists a.a\fresh \vec{x}$ of
  $CF_4$, the derivation is of the form
  \[\infer[\allR]{\cdot;\cdot\seq \forall \vec{x}{:}\vec{\tau}.\exists a{:}\nu.a\fresh \vec{x}}
  {\infer[F]{\vec{x}{:}\vec{\tau};\cdot \seq \exists a{:}\nu.a \fresh \vec{x}}
    {\infer[\exR,\Sigma\#]{\vec{x}{:}\vec{\tau}\#\Aa{:}\nu;\cdot \seq \exists a{:}\nu.a\fresh \vec{x}}
      {\hyp{\vec{x}{:}\vec{\tau}\#\Aa{:}\nu:\Aa \fresh \vec{x} \seq \Aa \fresh
          \vec{x}}}}}
  \]

  For a translated instance of $(CQ)$ of the form $\forall \vec{x}.(\new
  \Aa{:}\nu.\phi(\Aa,\vec{x}) \iff \exists a{:}\nu.a \fresh \vec{x} \andd
  \phi(a,\vec{x}))$, we will prove the two directions individually.
  For the forward direction, after some syntax-directed applications
  of right-rules we have
  \[
  \infer[\allR^n,\impR]{\cdot;\cdot \seq\forall \vec{x}{:}\vec{\tau}.(\new \Aa{:}\nu.\phi(\Aa,\vec{x})
    \impp \exists a{:}\nu.a \fresh \vec{x} \andd \phi(a,\vec{x}))}
  {\infer[\newL,\exR]{\vec{x}{:}\vec{\tau};\new \Aa{:}\nu.\phi(\Aa,\vec{x}) \seq \exists
      a.a \fresh \vec{x} \andd \phi(a,\vec{x})}
    {\infer[\andR]{\vec{x}{:}\vec{\tau}\#\Aa{:}\nu;\phi(\Aa,\vec{x}) \seq \Aa \fresh
        \vec{x} \andd \phi(\Aa,\vec{x})}
      {\infer[\Sigma\#^n]{\vec{x}{:}\vec{\tau}\#\Aa{:}\nu;\phi(\Aa,\vec{x})\seq \Aa \fresh \vec{x}}
        {\hyp[hyp]{\vec{x}{:}\vec{\tau}\#\Aa{:}\nu;\phi(\Aa,\vec{x}),\Aa \fresh \vec{x} \seq \Aa \fresh
            \vec{x}}} & \hyp{\vec{x}{:}\vec{\tau}\#\Aa{:}\nu;\phi(\Aa,\vec{x}) \seq
          \phi(\Aa,\vec{x})}}}}\]

  For the reverse direction, we need to show $\forall \vec{x}. \exists
  a{:}\nu.a \fresh \vec{x} \andd \phi(a,\vec{x}) \impp \new
  \Aa{:}\nu.\phi(\Aa,\vec{x})$.
  \[
  \infer[\allR,\impR]{\cdot;\cdot\seq \forall \vec{x}{:}\vec{\tau}.(\exists a{:}\nu. a \fresh
    \vec{x} \andd \phi(a,\vec{x}) \impp \new \Aa{:}\nu.\phi(\Aa,\vec{x}))}
  {\infer[\exL,\andL]{\vec{x}{:}\vec{\tau};\exists a{:}\nu. a \fresh \vec{x} \andd
      \phi(a,\vec{x}) \seq \new \Aa{:}\nu.\phi(\Aa,\vec{x})}
    {\infer[\newR]{\vec{x}{:}\vec{\tau},a{:}\nu; a \fresh \vec{x}, \phi(a,\vec{x}) \seq
        \new \Aa{:}\nu.\phi(\Aa,\vec{x})}
      {\infer[\Sigma\#^*,EVL]{\vec{x}{:}\vec{\tau},a{:}\nu\#\Ab{:}\nu; a \fresh \vec{x},
          \phi(a,\vec{x}) \seq \phi(\Ab,\vec{x})}
        {\infer[Ax^*]{\vec{x}{:}\vec{\tau},a{:}\nu\#\Ab{:}\nu; a \fresh \vec{x},\Ab \fresh
            \vec{x}, \swap{a}{\Ab}\phi(a,\vec{x}) \seq
            \phi(\Ab,\vec{x})}{\hyp{\vec{x}{:}\vec{\tau},a{:}\nu\#\Ab{:}\nu;
              \phi(\Ab,\vec{x}) \seq \phi(\Ab,\vec{x})}}
        }}}}
  \]
  Since both $a$ and $\Ab$ are fresh for all the other free variables
  of $\phi$, we have $\phi(a,\vec{x}) \iff
  \phi(\swap{\Ab}{a}{a},\swap{\Ab}{a}\vec{x}) \iff \phi(\Ab,\vec{x})$
  using equivariance and the fact that $a \fresh x\andd \Ab \fresh x
  \impp \swap{a}{\Ab}{x} \eq x$.

  Consequently, all the translations of axioms of $NL$ can be derived
  in \NLseq.  As a result, if $\Gamma' \subset NL$ is a finite set of
  axioms such that $\nd_{\NLseq} \Sigma;\Gamma,\Gamma'\seq\Delta$,
  then using the derivations of the axioms and finitely many instances
  of $cut$, we can obtain a derivation of $\nd_{\NLseq}
  \Sigma;\Gamma\seq\Delta$.
\end{proof}

Observe that this means that any closed theorem of $\NL$ can be
derived in \NLseq.  For example, from Pitts~\cite[Prop. 3 and 4]{pitts03ic}
we can show:
\begin{proposition}\labelProp{freshness-new-nl}
  \begin{itemize}
  \item If $FV(t) \subseteq \vec{x}$ and $FN(t) = \emptyset$ then we
    can derive
    $\Sigma;\Gamma \seq \forall a{:}\nu. \forall
    \vec{x}{:}\vec{\tau}. a \fresh x_1 \andd \cdots \andd a \fresh x_n
    \impp a \fresh t$.
\item If $FV(\phi) \subseteq \{a,\vec{x}\}$ then we can derive
  $\Sigma;\Gamma \seq \exists a{:}\nu. a \fresh \vec{x} \andd
  \phi(a,\vec{x}) \iff \forall a{:}\nu. a \fresh \vec{x} \impp \phi(a,\vec{x})$
  \end{itemize}
\end{proposition}

Now we consider the converse: showing that there are
no ``new theorems'', that any $\NL$ sequent derivable in \NLseq is also
derivable in $\NL$.  This is not as straightforward because
subderivations of translated \NL judgments may involve name-symbols.
However, we can show that such name-symbols can always be removed.

We also introduce a converse translation mapping \NLseq formulas to
\NL formulas:
\begin{eqnarray*}
A^\dagger &= &A\\
\false^\dagger &=& \false\\
(\phi \impp \psi)^\dagger&=& \phi^\dagger \impp \psi^\dagger\\
(\forall x{:}\tau.\phi)^\dagger &=& \forall x{:}\tau.\phi^\dagger\\
(\new \Aa{:}\nu.\phi)^\dagger &=& \new a{:}\nu. (\phi^\dagger[a/\Aa]) \quad
(\iota(a) = \Aa)
\end{eqnarray*}
Technically, we translate \NL formulas over variables $\Var$ to \NLseq
formulas over $\Var \cup \Name'$, again using the bijection $\iota$  between name-variables $\Name'$
and names in $\Name$.  Note that
(up to $\alpha$-equivalence) the $(-)^*$-translation and
$(-)^\dagger$-translation are inverses.  We also define
the set $\|\Sigma\|$ as follows:
\[\|\Sigma\| = \{a \fresh x \mid \iota(a) \fresh x \in
|\Sigma|\} \cup \{a \fresh b \mid \iota(a) \fresh \iota(b) \in
|\Sigma|\} \]
that is, $\|\Sigma\|$ is the finite subset of $|\Sigma|$ consisting of
constraints whose right-hand sides are variables or names,
but with names  replaced by the corresponding name-variables according
to the bijection $\iota$.  

We can now show the desired result.  
\begin{theorem}\labelThm{conservativity-if}
  If $\Sigma;\Gamma \seq \Delta$ is
  derivable in \NLseq then $\Sigma^\dagger;\Gamma^\dagger,\|\Sigma\|\seq\Delta^\dagger$ is derivable in
  \NL.
\end{theorem}
\begin{proof}
  The proof is by induction on the logical height of this
  derivation, with secondary induction on the total height.  For the cases corresponding to first-order/equational
  proof rules, the induction step is straightforward.

  For the cases corresponding to nonlogical rules corresponding to
  universal axioms $\forall \vec{x}.\Andd \vec{P} \impp \Orr\vec{Q}$,
  suppose that we have derivations of the form
  \[\infer[Ax]{\Sigma;\Gamma,\vec{P} \seq \Delta}{\Sigma ;
    \Gamma,\vec{P},Q_1 \seq \Delta & \Sigma ;\Gamma,\vec{P},Q_n \seq
    \Delta}\]
  Then by induction, we have $NL$ derivations of the $NL$ sequents
  $\Sigma^\dagger ; \Gamma^\dagger,\vec{P},Q_i,\|\Sigma\| \seq
  \Delta^\dagger$ for $i \in \{1,\ldots,n\}$.  It is straightforward
  to show that each of the
  axioms in \refFig{nonlogical-axioms} is provable in $\NL$, hence we
  can cut against each axiom instance:
  \[\infer[cut]{\Sigma^\dagger;\Gamma^\dagger,\vec{P},\|\Sigma\|\seq
    \Delta^\dagger}{
\Sigma^\dagger;\cdot \seq\forall
    \vec{x}.\Andd \vec{P} \impp \Orr\vec{Q} & \infer[\allR]{\Sigma^\dagger; \Gamma^\dagger,\vec{P},\forall
      \vec{x}.\Andd \vec{P} \impp \Orr\vec{Q} ,\|\Sigma\|\seq \Delta^\dagger}{
      \infer[\impL]{\Sigma^\dagger;\Gamma^\dagger,\vec{P},\Andd \vec{P} \impp \Orr\vec{Q},\|\Sigma\|\seq
        \Delta^\dagger}{\Sigma^\dagger;\Gamma^\dagger,\vec{P} ,\|\Sigma\|\seq \Andd \vec{P} & \infer[\orL^n]{\Sigma^\dagger
          ;\Gamma^\dagger,\vec{P},\Orr \vec{Q},\|\Sigma\| \seq \Delta^\dagger}
        {\Sigma^\dagger;\Gamma^\dagger,\vec{P},Q_1 ,\|\Sigma\|\seq\Delta^\dagger & \cdots &
          \Sigma^\dagger;\Gamma^\dagger,\vec{P},Q_n ,\|\Sigma\|\seq\Delta^\dagger }}}}\]

  The cases for $F_3,F_4,F,A_2, A_3, \Sigma\#,\newL,\newR$ remain.

  For $F_3$, we have a derivation 
  \[\infer[F_3]{\Sigma;\Gamma,a \fresh a \seq \Delta}{}\]
  In $\NL$ we can derive $ \Sigma^\dagger;\Gamma^\dagger,a \fresh a,\|\Sigma\| \seq \Delta^\dagger$ using
  the $a \fresh b \impp a \not\eq b$ direction of $(CF_2)$ since $a
  \not\eq a$ is contradictory.

  For $F_4$, we have a derivation
  \[\infer[F_4]{\Sigma;\Gamma\seq \Delta}{\Sigma;\Gamma, a \eq b \seq
    \Delta&\Sigma;\Gamma, a \fresh b\seq \Delta}\]
  By induction, we have derivations of $\Sigma^\dagger;\Gamma^\dagger,
  a \eq b,\|\Sigma\| \seq \Delta^\dagger$ and
  $\Sigma^\dagger;\Gamma^\dagger, a \fresh b,\|\Sigma\| \seq
  \Delta^\dagger$.  Since $a \fresh b \iff a \not\eq b$ and $a \eq b
  \orr a \not\eq b$ is a tautology in classical logic, $a \fresh b
  \orr a \not\eq b$ is also a tautology.  We can cut against a
  derivation of this formula to derive $\Sigma;\Gamma\seq \Delta$ in
  $NL$.

  For $F$, suppose we have a derivation of the form
  \[\infer[F]{\Sigma;\Gamma \seq \Delta}{\Sigma\#\Aa{:}\nu ;\Gamma\seq\Delta}\]
By induction, we can derive the $NL$ sequent $\Sigma^\dagger,a{:}\nu
;\Gamma^\dagger,\|\Sigma\#\Aa{:}\nu\|\seq\Delta^\dagger$.  Note that
$\|\Sigma\#\Aa{:}\nu\| = \|\Sigma\|, a\fresh \vec{x}$ where
$\vec{x}= FV(\Sigma^\dagger)$.  
Using the freshness axiom $(CF_4)$ of \NL, we can derive
  \[\infer[cut]{\Sigma^\dagger;\Gamma^\dagger,\|\Sigma\| \seq
    \Delta^\dagger}{
\Sigma^\dagger;\cdot\seq \forall
    \bar{x}{:}\vec{\tau}.\exists a{:}\nu.a \fresh \bar{x} &
    \infer[\allL,\exL]{\Sigma^\dagger;\Gamma^\dagger,\forall \bar{x}.\exists a{:}\nu.
      a \fresh \vec{x} ,\|\Sigma\| \seq \Delta^\dagger}{\Sigma^\dagger,a{:}\nu;\Gamma^\dagger,\|\Sigma\| ,a \fresh \vec{x}
      \seq \Delta^\dagger}}\]

  It is likewise easy to derive
  rules $A_2,A_3$ from axioms $(CA_1),(CA_2)$ of $\NL$ using $cut$.
  
  For $\Sigma\#$, suppose we have a derivation of the form:
\[
\infer[\Sigma\#]{\Sigma_1\#\Aa{:}\nu,\Sigma_2;\Gamma\seq \Delta}{
\Sigma_1\#\Aa{:}\nu,\Sigma_2;\Gamma,\Aa \fresh t\seq \Delta 
&
(\Aa \fresh t \in  |\Sigma_1|)}\]
By induction, we have
$\Sigma_1^\dagger,a{:}\nu,\Sigma_2^\dagger;\Gamma^\dagger,\Aa\fresh t,
\|\Sigma_1\#\Aa{:}\nu,\Sigma_2\|\seq \Delta^\dagger$.  Observe that $a
\fresh \vec{x} \subseteq \|\Sigma_1 \#\Aa{:}\nu,\Sigma_2\|$.  
Using \refProp{freshness-new-nl}(1), we can derive as follows:
\[
\infer[cut]{\Sigma_1^\dagger,a{:}\nu,\Sigma_2^\dagger;\Gamma^\dagger,
\|\Sigma_1\#\Aa{:}\nu,\Sigma_2\|\seq \Delta^\dagger}{
\Sigma_1^\dagger,a{:}\nu,\Sigma_2^\dagger;\cdot \seq \forall
a{:}\nu.\forall \vec{x}{:}\vec{\tau}. a \fresh \vec{x} \impp a \fresh
t& 
\infer[\allL^*,\impL^*]{\Sigma_1^\dagger,a{:}\nu,\Sigma_2^\dagger;\Gamma^\dagger,
\|\Sigma_1\#\Aa{:}\nu,\Sigma_2\|,\forall
a{:}\nu.\forall \vec{x}{:}\vec{\tau}. a \fresh \vec{x} \impp a \fresh
t\seq \Delta^\dagger
}{
\Sigma_1^\dagger,a{:}\nu,\Sigma_2^\dagger;\Gamma^\dagger,
\|\Sigma_1\#\Aa{:}\nu,\Sigma_2\|, a \fresh
t\seq \Delta^\dagger
}
}
\]

  Finally, we consider the cases for $\newL$ and $\newR$.
  For $\newL$, we have
  \[\infer[\newL]{\Sigma;\Gamma,\new \Aa{:}\nu.\phi(\Aa,\vec{x}) \seq
    \Delta}{\Sigma\#\Aa{:}\nu;\Gamma,\phi(\Aa,\vec{x}) \seq \Delta}\]
  From the upper derivation, by induction, we have a derivation of
  $\Sigma^\dagger,a{:}\nu ;
  \Gamma^\dagger,\|\Sigma\#\Aa{:}\nu\|,\phi^\dagger(a,\vec{x}) \seq
  \Delta^\dagger$. Since $\|\Sigma\#\Aa{:}\nu\| = \|\Sigma\|, a \fresh
  \vec{y}$ where $\vec{y} = FV(\Sigma^\dagger) \supseteq \vec{x}$, we can also derive $\Sigma^\dagger;\Gamma^\dagger,\|\Sigma\|,\exists a{:}\nu. a \fresh \vec{x} \andd
  \phi(a,\vec{x}) \seq \Delta$ using $\exL$ and $\allL$.  Finally, we
  can cut against the axiom instance $\forall
  \vec{x}{:}\vec{\tau}.\exists a{:}\nu.a \fresh\vec{x} \andd
  \phi^\dagger(a,\vec{x}) \iff \new a{:}\nu.\phi^\dagger(a,\vec{x})$ to prove that
  $\Sigma^\dagger;\Gamma^\dagger,\new a{:}\nu.\phi^\dagger(a,\vec{x}) \seq \Delta^\dagger$.

  For $\newR$, we have 
  \[\infer[\newR]{\Sigma;\Gamma \seq
    \new \Aa{:}\nu.\phi(\Aa,\vec{x}),\Delta}{\Sigma\#\Aa{:}\nu;\Gamma \seq \phi(\Aa,\vec{x}),\Delta}\]
  The argument is similar to the previous case: by induction, we
  can derive $\Sigma^\dagger,a{:}\nu ;
  \Gamma^\dagger,\|\Sigma\#\Aa{:}\nu\| \seq
  \phi^\dagger(a,\vec{x}),\Delta^\dagger$ in $NL$.  Thus, since
  $\|\Sigma\#\Aa{:}\nu\| = \|\Sigma\|,a \fresh \vec{y}$ where $\vec{y}
  = FV(\Sigma^\dagger)$, we can
  conclude $\Sigma^\dagger;\Gamma^\dagger,\|\Sigma\| \seq \forall a{:}\nu. a \fresh
  \vec{y} \impp \phi(a,\vec{x}),\Delta^\dagger$.
 Using \refProp{freshness-new-nl}(2) and the axiom $(CQ)$
  defining $\new$ in \NL we can cut against the formula
  \[\forall \vec{y}{:}\vec{\tau}.(\forall a{:}\nu.a\fresh \vec{y} \impp \phi^\dagger(a,\vec{x})) \iff \new
  a{:}\nu.\phi^\dagger(a,\vec{x})\]
where $\vec{y} \supseteq \vec{x}$.  
 We can conclude that $\Sigma^\dagger;\Gamma^\dagger,\|\Sigma\| \seq
  \Delta^\dagger,\new a{:}\nu.\phi^\dagger(a,\vec{x})$.
\end{proof}

\begin{corollary}
  If $\Sigma$ only contains variables and $\nd_{\NL}\Sigma;\Gamma^* \seq
  \Delta^*$ then $\Sigma;\Gamma \seq \Delta$ is derivable in \NLseq.
\end{corollary}
\begin{proof}
  By \refThm{conservativity-if}, we know that
  $\Sigma^\dagger;(\Gamma^*)^\dagger,\|\Sigma\| \seq
  (\Delta^*)^\dagger$.  By definition of the $(-)^*$ and $(-)^\dagger$
  translations, we know that $
  (\Gamma^*)^\dagger = \Gamma$ and $(\Delta^*)^\dagger = \Delta$.
  Moreover, since $\Sigma$ contains no name-symbols, by definition
  $\Sigma^\dagger = \Sigma$ and $\|\Sigma\| = \emptyset$. Hence,
  $\Sigma;\Gamma \seq \Delta$.
\end{proof}

\subsubsection{Intuitionistic Nominal Logic}
\labelSec{intuitionistic-conservativity}

\begin{figure*}[tb]
  \[
  \begin{array}{rc}
    \multicolumn{2}{l}{\text{Swapping}}\\
    (IS_1)& \forall a{:}\nu, x{:}\tau.\; (a~a)\act x \eq x\\
    (IS_2)& \forall a,a'{:}\nu,x{:}\tau.\; (a~a')\act(a~a')\act x \eq x\\
    (IS_3)& \forall a,a'{:}\nu.\; (a~a')\act a \eq a'\\
    \multicolumn{2}{l}{\text{Equivariance}}\\
    (IE_1)& \forall a,a'{:}\nu,b,b'{:}\nu',x{:}\tau.\; (a~a')\act (b~b')\act x \eq ((a~a')\act b~(a~a')\act b')\act(a~a')\act x\\
    (IE_2)& \forall a,a'{:}\nu,b{:}\nu',x{:}\tau.\; b\fresh x \impp (a~a')\act b \fresh (a~a') \act x\\
    (IE_3)& \forall a,a'{:}\nu,\vec{x}:\vec{\tau}.\; (a~a')\act f(\bar{x}) \eq f((a~a')\act \vec{x}) \\
    (IE_4)& \forall a,a'{:}\nu,\vec{x}:\vec{\tau}.\; p(\vec{x}) \impp p((a~a')\act \vec{x}) \\
    (IE_5) & \forall b,b'{:}\nu', a{:}\nu, x{:}\tau.\; (b~b')\act (\abs{a}x) \eq \abs{(b~b')\act a} ((b~b')\act x)\\
    \multicolumn{2}{l}{\text{Freshness}} \\
    (IF_1) & \forall a,a'{:}\nu, x{:}\tau.\; a\fresh x \andd a' \fresh x \impp (a~a')\act x \eq x\\
    (IF_2) & \forall a{:}\nu.\; \nott(a \fresh a)\\
    (IF_3) & \forall a,a'{:}\nu.\; a \fresh a' \orr a \eq a'\\
    (IF_4) & \forall a{:}\nu,a'{:}\nu'.\; a\fresh a'\\
    (IF_5) & \forall \vec{x}:\vec{\tau}.\; \exists a{:}\nu.~ a \fresh \vec{x}\\
    \multicolumn{2}{l}{\text{$\new$-quantifier}} \\
    (IQ)& \forall \vec{x}. (\new a{:}\nu.\;\phi) \iff (\exists a{:}\nu.\; a\fresh \vec{x} \andd \phi)\\
    \multicolumn{2}{l}{\text{where $FV(\new a.\phi) \subseteq \{\vec{x}\}$}}\\
    \multicolumn{2}{l}{\text{Abstraction}} \\
    (IA_1)& \begin{array}{rcl}
      \forall a,a'{:}\nu, x,x'{:}\tau.\; \abs{a}x \eq \abs{a'}x' &\iff& (a \eq a' \andd x \eq x')\\
      &\orr& (a' \fresh x \andd x' \eq (a~a')\act x)
    \end{array}\\
    (IA_2)& \forall y:\abs{\nu}{\tau}.\exists a{:}\nu,x{:}\tau.~ y \eq
    \abs{a}{x}
  \end{array}
  \]
\hrule
  \caption{Axioms of Intuitionistic Nominal Logic}\labelFig{inl-ax}
\end{figure*}

We wish to argue that the intuitionistic calculus \INLseq is really
``intuitionistic nominal logic''.  However, Pitts only considered
classical nominal logic.  There is a subtlety having to do with Pitts'
axiom $(CF_2)$ in the intuitionistic case.

Pitts' original axiom $(CF_2)$ stated that freshness among names is the
same as inequality:
\[(CF_2)\quad \forall a,a'{:}\nu.\; a \fresh a' \iff \nott(a \eq a')\] 
However, this axiom does not fit the scheme for nonlogical rules given
by Negri and von Plato~\cite{negri01structural}.
Instead, in \INLseq we use two nonlogical rules $F_3$ and $F_4$
asserting that no name is fresh for itself and that two names (of the
same type) are
either equal or fresh.  These two axioms are equivalent to $(CF_2)$ in
classical logic, but in intuitionistic logic, Pitts' axiom is weaker,
since $a \eq b \orr a \not\eq b$ does not follow from $(CF_2)$.  (Recall that
for the $F_4$ case of \refThm{conservativity-if}, we used excluded
middle for name-equality).

We have modified Pitts' axiomatization slightly by replacing the
original axiom $(CF_2)$ with two rules, $(IF_2)$ asserting that no
name is fresh for itself, and $(IF_3)$ stating that two names are
either fresh or equal.  In classical logic, these are equivalent
axiomatizations, whereas $(IF_3)$ is not provable in intuitionistic
logic from Pitts' axioms.  Moreover, it is computationally plausible that
equality and freshness among names are both decidable, since names are
typically finite, discrete data structures.

For this reason, we introduce an alternative axiomatization $INL$,
shown in \refFig{inl-ax}, differing in the replacement of $(CF_2)$
with two axioms $(IF_2)$ and $(IF_3)$.  These axioms are equivalent in
classical logic to $(CF_2)$, but better-behaved from a proof-theoretic 
perspective.

Let $\nd_{INL}$ indicate derivability in intuitionistic logic from the
axioms in $INL$. Using essentially the same proof techniques as for the
classical case, we have:
\begin{theorem}
  If $\Sigma$ contains only variables, then $\nd_{INL}\Sigma;\Gamma
  \seq \Delta$ is derivable if and only if $\nd_{\INLseq}
  \Sigma;\Gamma^*\seq\Delta^*$.
\end{theorem}

\section{Conclusions}\labelSec{concl}

Nominal logic provides powerful techniques for reasoning about fresh
names and name-binding.  One of the most interesting features of
nominal logic is the $\new$-quantifier.  However, the techniques used
for reasoning with $\new$ offered by previous formalizations of
nominal logic are highly (but unnecessarily) complex.

In this article we have introduced a new sequent calculus $\NLseq$ for nominal
logic which uses typing contexts extended with freshness information to deal with
the $\new$-quantifier.  Its rules for $\new$ are symmetric and
rationalize a proof-search semantics for $\new$ that seems natural and
intuitive (inspired by the treatment of $\new$ in nominal logic
programming).  We proved cut-elimination in detail.  In addition, we
used \NLseq to provide a syntactic proof of consistency and a detailed
proof of equivalence to Pitts' axiomatization modulo ordinary
first-order (classical/intuitionistic) logic.  These results are the
first of their kind to be shown in detail.

\NLseq has also been used in other work:
\begin{itemize}
\item \NLseq provides a proof-search reading of $\new$ which is much
  closer to the approach taken in the \aprolog nominal logic
  programming language~\cite{cheney04iclp,cheney08toplas}.  While
  Gabbay and Cheney gave a proof-theoretic semantics of nominal logic
  programming based on \FLseq, this analysis does not seem relevant to
  \aprolog because it suggests a quite different (and, for typical
  programs, much more computationally intensive) proof-search
  technique for $\new$-quantified formulas.  In contrast, \NLseq seems
  to provide a proof-theoretic foundation for \aprolog's existing
  search technique. 
\item Gabbay and Cheney~\cite{gabbay04lics} showed that \FOLNabla,
  another logic due to Miller and Tiu~\cite{miller05tocl} possessing a
  self-dual ``fresh value'' quantifier, can be soundly interpreted in
  a higher-order variant of \FLseq via a proof-theoretic translation.
  However, the translation they developed was incomplete, and the
  possibility of finding a faithful translation was left open.
  Cheney~\cite{cheney05fossacs} showed how to translate to a
  higher-order variant of \NLseq and proved a completeness result.  In
  this paper we have focused on \NLseq only over first-order terms.
  It would be interesting to further explore \NLseq over higher-order
  terms and compare its expressiveness to more recent variations of
  Miller and Tiu's approach, such as the ``nominal
  abstraction'' system of Gacek et al.~\cite{gacek11ic}.
\item Miculan, Scagnetto and Honsell~\cite{miculan05merlin} have shown
  how to translate derivable judgments from (a natural-deduction
  variant of) \NLseq to the Theory of Contexts~\cite{honsell01tosca},
  an extension of the Calculus of Inductive Constructions with a
  theory axiomatizing a type of names with decidable equality,
  freshness, and name-binding encoded as second-order function
  symbols.  It may be interesting to consider the reverse direction,
  e.g. translating a first-order fragment of the Theory of Contexts to nominal logic.
\end{itemize}

Additional directions for future work include the development of
natural deduction calculi and type theories using the ideas of \NLseq.
One particularly interesting direction is the possibility of
developing a type system and confluent term rewriting system that
could be used to decide equality of nominal terms and proof terms.  In
such a system, the explicit equality and freshness theory that
necessitates the many nonlogical rules in \NLseq could be dealt with
implicitly via traditional rewriting and syntactic side-conditions,
leading to an even simpler proof theory for nominal logic.  However,
work in this direction by Sch\"opp and Stark~\cite{schoepp04csl}
indicates that there may be significant obstacles to this approach;
the system introduced in this article may be viewed as a well-behaved
fragment of their system.  Further development of the proof theory and
type theory of nominal logic (for example, building on nominal type
theories by Pitts~\cite{pitts11jfp}, Cheney~\cite{cheney12lmcs}, or
Crole and Nebel~\cite{crole13mfps}) seems possible and desirable.

\bibliographystyle{plain} 
\bibliography{nominal,ac,logicprog,logic,paper}

\end{document}

%% file: nominal.tex
\usepackage{xspace}
\usepackage{graphicx}

\usepackage{amssymb}
\usepackage{amsfonts}
\usepackage{proof}
\usepackage{epsfig}
\usepackage{verbatim}

\renewcommand{\emptyset}{\varnothing}
\renewcommand{\models}{\vDash}

\newcommand{\aprolog}{{\ensuremath \alpha}{Pro\-log}\xspace}

\newcommand{\ab}[1]{\langle #1 \rangle}

\newcommand{\NLseq}{\ensuremath{NL^\seq}\xspace}
\newcommand{\INLseq}{\ensuremath{INL^\seq}\xspace}

\newcommand{\Andd}{\bigwedge}
\newcommand{\Orr}{\bigvee}
\newcommand{\andd}{\wedge}
\newcommand{\orr}{\vee}
\newcommand{\impp}{\supset}

\newcommand{\nott}{\neg}
\newcommand{\true}{\top}
\newcommand{\false}{\bot}
\newcommand{\eq}{\approx}

\newcommand{\trueR}{{\top}R}
\newcommand{\falseL}{{\bot}L}

\newcommand{\impL}{{\impp}L}
\newcommand{\impR}{{\impp}R}
\newcommand{\andL}{{\andd}L}
\newcommand{\andR}{{\andd}R}
\newcommand{\orL}{{\orr}L}
\newcommand{\orR}{{\orr}R}
\newcommand{\allL}{{\forall}L}
\newcommand{\allR}{{\forall}R}
\newcommand{\exL}{{\exists}L}
\newcommand{\exR}{{\exists}R}
\newcommand{\newL}{{\new}L}
\newcommand{\newR}{{\new}R}

\newcommand{\hyp}[2][]{\infer[#1]{#2}{}}

\newcommand{\new}[1][]{\reflectbox{\sf{#1{}N}}}

\newcommand{\fresh}{\mathrel{\#}}

\newcommand{\abs}[2]{{\ab{#1}{#2}}}

\newcommand{\name}[1]{\mathsf{#1}}
\newcommand{\Aa}{\name{a}}
\newcommand{\Ab}{\name{b}}
\newcommand{\Ac}{\name{c}}

\newcommand{\swap}[3]{(#1~#2)\act#3}

\newcommand{\act}{\boldsymbol{\cdot}}

\newcommand{\labelSec}[1]{\label{sec:#1}}
\newcommand{\refSec}[1]{Section~\ref{sec:#1}}
\newcommand{\labelFig}[1]{\label{fig:#1}}
\newcommand{\refFig}[1]{Figure~\ref{fig:#1}}

\newcommand{\labelThm}[1]{\label{thm:#1}}
\newcommand{\refThm}[1]{Theorem~\ref{thm:#1}}
\newcommand{\labelLem}[1]{\label{lem:#1}}
\newcommand{\refLem}[1]{Lemma~\ref{lem:#1}}
\newcommand{\labelProp}[1]{\label{prop:#1}}
\newcommand{\refProp}[1]{Proposition~\ref{prop:#1}}

\newcommand{\subs}[2]{\{#1/#2\}}

\newcommand{\seq}{\Rightarrow}
\newcommand{\nd}{\vdash}

\newcommand{\FOLNabla}{\ensuremath{FO\lambda^\nabla}\xspace}


%% file: paper.bbl
\begin{thebibliography}{10}

\bibitem{cheney05fossacs}
J.~Cheney.
\newblock A simpler proof theory for nominal logic.
\newblock In {\em FOSSACS 2005}, volume 3441 of {\em LNCS}, pages 379--394.
  Springer-Verlag, 2005.

\bibitem{cheney06jsl}
J.~Cheney.
\newblock Completeness and {Herbrand} theorems for nominal logic.
\newblock {\em Journal of Symbolic Logic}, 71(1):299--320, 2006.

\bibitem{cheney04iclp}
J.~Cheney and C.~Urban.
\newblock {Alpha}-{Prolog}: A logic programming language with names, binding
  and alpha-equivalence.
\newblock In {\em Proceedings of the 20th International Conference on Logic
  Programming ({ICLP} 2004)}, number 3132 in LNCS, pages 269--283, St. Malo,
  France, 2004. Springer-Verlag.

\bibitem{cheney12lmcs}
James Cheney.
\newblock A dependent nominal type theory.
\newblock {\em Logical Methods in Computer Science}, 8(1), 2012.

\bibitem{cheney08toplas}
James Cheney and Christian Urban.
\newblock Nominal logic programming.
\newblock {\em ACM Transactions on Programming Languages and Systems},
  30(5):26, August 2008.

\bibitem{crole13mfps}
Roy~L. Crole and Frank Nebel.
\newblock Nominal lambda calculus: An internal language for {FM}-cartesian
  closed categories.
\newblock In {\em MFPS}, 2013.
\newblock In press.

\bibitem{gabbay04lics}
M.~J. Gabbay and J.~Cheney.
\newblock A sequent calculus for nominal logic.
\newblock In {\em LICS 2004}, pages 139--148. IEEE, 2004.

\bibitem{gabbay02fac}
M.~J. Gabbay and A.~M. Pitts.
\newblock A new approach to abstract syntax with variable binding.
\newblock {\em Formal Aspects of Computing}, 13:341--363, 2002.

\bibitem{gabbay07jal}
Murdoch Gabbay.
\newblock Fresh logic: proof-theory and semantics for {FM} and nominal
  techniques.
\newblock {\em J. Applied Logic}, 5(2):356--387, 2007.

\bibitem{gacek11ic}
Andrew Gacek, Dale Miller, and Gopalan Nadathur.
\newblock Nominal abstraction.
\newblock {\em Inf. Comput.}, 209(1):48--73, 2011.

\bibitem{honsell01tosca}
Furio Honsell, Marino Miculan, and Ivan Scagnetto.
\newblock The theory of contexts for first order and higher order abstract
  syntax.
\newblock In {\em TOSCA}, volume~62 of {\em Electronic Notes on Theoretical
  Computer Science}, 2001.

\bibitem{miculan05merlin}
M.~Miculan, I.~Scagnetto, and F.~Honsell.
\newblock Translating specifications from nominal logic to {CIC} with the
  theory of contexts.
\newblock In R.~Pollack, editor, {\em MERLIN}, pages 41--49, Tallinn, Estonia,
  September 2005. ACM Press.

\bibitem{miller05tocl}
Dale Miller and Alwen Tiu.
\newblock A proof theory for generic judgments.
\newblock {\em ACM Trans. Comput. Logic}, 6(4):749--783, 2005.

\bibitem{negri01structural}
Sara Negri and Jan {von Plato}.
\newblock {\em Structural Proof Theory}.
\newblock Cambridge University Press, 2001.

\bibitem{pitts03ic}
A.~M. Pitts.
\newblock Nominal logic, a first order theory of names and binding.
\newblock {\em Information and Computation}, 183:165--193, 2003.

\bibitem{pitts11jfp}
A.~M. Pitts.
\newblock Structural recursion with locally scoped names.
\newblock {\em Journal of Functional Programming}, 21(3):235--286, 2011.

\bibitem{schoepp04csl}
Ulrich Sch\"opp and Ian Stark.
\newblock A dependent type theory with names and binding.
\newblock In {\em CSL 2004}, number 3210 in LNCS, pages 235--249, Karpacz,
  Poland, 2004.

\bibitem{troelstra00basic}
A.~S. Troelstra and H.~Schwichtenberg.
\newblock {\em Basic Proof Theory}.
\newblock Number~43 in Cambridge Tracts in Theoretical Computer Science.
  Cambridge University Press, second edition, 2000.

\bibitem{urban04tcs}
C.~Urban, A.~M. Pitts, and M.~J. Gabbay.
\newblock Nominal unification.
\newblock {\em Theoretical Computer Science}, 323(1--3):473--497, 2004.

\end{thebibliography}
